\documentclass[11pt]{article}

\usepackage{a4wide}
\usepackage{amsmath,amssymb,amsthm,amscd}
\usepackage{graphicx}
\usepackage{authblk}

\usepackage{natbib}
\usepackage{algorithm}	
\usepackage{algorithmic}
\usepackage{bm}
\usepackage{bbm}
\usepackage{enumitem}

\usepackage[toc,page]{appendix}

\graphicspath{ {Images/} }

\theoremstyle{plain}

\newtheorem{Theorem}{Theorem}		

\newtheorem{Proposition}{Proposition}

\newtheorem{Assumption}{Assumption}

\newtheorem*{Proof}{Proof}

\providecommand{\keywords}[1]{\textbf{\textit{Keywords: }} #1}

\begin{document}

\title{An Extended Empirical Saddlepoint Approximation for Intractable Likelihoods}

\author[1,$\dag$]{Matteo Fasiolo}
\author[1]{Simon N. Wood}
\author[2]{Florian Hartig}
\author[3]{Mark V. Bravington} 
\affil[1]{School of Mathematics, University of Bristol}
\affil[2]{Department of Biometry and Environmental System Analysis, University of Freiburg}
\affil[3]{CSIRO Center for Mathematics and Information Science}
\affil[$\dag$]{Correspondence: matteo.fasiolo@bristol.ac.uk}

\maketitle

\abstract{
The challenges posed by complex stochastic models used in computational ecology, biology and genetics have stimulated the development of approximate approaches to statistical inference. Here we focus on Synthetic Likelihood (SL), a procedure that reduces the observed and simulated data to a set of summary statistics, and quantif{\text{}}ies the discrepancy between them through a synthetic likelihood function. SL requires little tuning, but it relies on the approximate normality of the summary statistics. We relax this assumption by proposing a novel, more f{\text{}}lexible, density estimator: the Extended Empirical Saddlepoint approximation. In addition to proving the consistency of SL, under either the new or the Gaussian density estimator, we illustrate the method using three examples. One of these is a complex individual-based forest model for which SL offers one of the few practical possibilities for statistical inference. The examples show that the new density estimator is able to capture large departures from normality, while being scalable to high dimensions, and this in turn leads to more accurate parameter estimates, relative to the Gaussian alternative. The new density estimator is implemented by the \verb|esaddle| R package, which can be found on the Comprehensive R Archive Network (CRAN).
}

\vspace{4pt}


\section{Introduction}

Synthetic Likelihood (SL) \citep{wood2010} is a simulation-based inferential procedure similar to Approximate Bayesian Computation (ABC) methods \citep{beaumont2002approximate}, but with the practical advantage of requiring much less tuning. In \cite{wood2010} tuning is avoided partly through a Gaussian assumption on the distribution of the statistics used to compare the data and the model output. This assumption can be problematic, as illustrated by the following very simple population dynamic model of an organism subject to boom and bust dynamics with stochastic external recruitment:
\begin{equation} \label{eq:crashModel}
N_{t+1} \sim
\begin{cases}
\text{Pois}\{N_t(1+r)\} + \epsilon_t, & \text{if } N_t \leq \kappa,
\\
\text{Binom}(N_t, \alpha) + \epsilon_t, & \text{if } N_t > \kappa,
\end{cases}
\end{equation}
where $\epsilon_t \sim \text{Pois}(\beta)$ is a stochastic arrival process, with rate $\beta > 0$, and $t=1, \dots, \text{T}$. The population $N_t$ grows stochastically at rate $r>0$, but it crashes if the carrying capacity $\kappa$ is exceeded. The severity of the crash depends on the survival probability $\alpha\in(0,1)$. Two fairly natural statistics for this model are the population mean and the number of periods during which $N_t \leq 1$, the latter being useful for identifying $\beta$. However, as Figure \ref{fig: crash_dist} shows, the distribution of these statistics is far from normal, which could affect the accuracy of the parameter estimates produced by SL. The purpose of this paper is to develop a version of SL that relaxes the normality requirement, while retaining the tuning free advantages of the original method. We do this by replacing the Gaussian assumption with a new density estimator: the Extended Empirical Saddlepoint (EES) estimator. We prove the consistency of the resulting parameter estimator, and illustrate that the method can yield substantial inferential improvements when multivariate Gaussianity is untenable. We also provide examples where ABC methods would require exceedingly low tolerances and low acceptance rates in order to achieve equivalent accuracy.

The most important commonality between SL and ABC methods is that both base statistical inference on a vector of summary statistics, ${ \bm s^0} = S({\bm y}^0)$, rather than on the full data, ${\bm y}^0$. However, while ABC methods explicitly aim at sampling from the approximate posterior $p(\bm \theta| \bm s^0)$, SL provides a parametric approximation to $p(\bm s^0|\bm \theta)$. This synthetic likelihood, which we indicate with $p_{SL}(\bm s^0|\bm \theta)$, can then be used within a Bayesian or a classical context. \cite{wood2010} used a multivariate Gaussian density to approximate the distribution of the summary statistics. Under this distributional assumption, a pointwise estimate of the synthetic likelihood at $\bm \theta$ can be obtained using Algorithm \ref{alPoint}.
\begin{algorithm}
\caption{Estimating $p_{SL}({ \bm s}^0|{ \bm \theta})$}
\label{alPoint}
\begin{algorithmic}[1]
\STATE  Simulate datasets ${ \bm Y}_i, \dots, {\bm Y}_m$ from the model $p({ \bm y}| \bm \theta)$.
\STATE  Transform each dataset ${ \bm Y}_i$ to a vector of summary statistics $ \bm S_i = S({ \bm Y}_i)$.
\STATE  Calculate sample mean $\hat{ \bm \mu}_{ \bm \theta}$ and covariance $\hat{ \bm \Sigma}_{ \bm \theta}$ of the simulated statistics, possibly robustly. 
\STATE  Estimate the synthetic likelihood
\begin{equation*}
\hat{p}_{SL}({ \bm s}^0|{ \bm \theta}) = (2\pi)^{-\frac{d}{2}}\text{det}\big(\hat{ \bm \Sigma}_{ \bm \theta}\big)^{-\frac{1}{2}} \exp \bigg \{-\frac{1}{2}({ \bm s}^0 - \hat { \bm \mu}_{ \bm \theta})^T \hat { \bm \Sigma}^{-1}_{ \bm \theta} ({ {\bm s}}^0 - \hat { \bm \mu}_{ \bm \theta}) \bigg \}, 
\end{equation*}
 where $d$ is the number of summary statistics used. 
\end{algorithmic}
\end{algorithm}

One advantage of SL, over most ABC methods, is that it does not require the user to choose a tolerance or an acceptance threshold and that the summary statistics are scaled automatically and dynamically by $\hat {\bm \Sigma}_{\bm \theta}$. In addition, \cite{blum2010approximate} showed that the convergence rate of ABC methods degrades rapidly with $d$. This curse of dimensionality, brought about by the non-parametric nature of ABC, forces practitioners to use dimension reduction or statistics selection techniques, such as those described by \cite{blum2013comparative}. SL is less sensitive to the number of statistics used, due to the parametric likelihood approximation. 

\begin{figure}
\centering
\includegraphics[scale = 0.40]{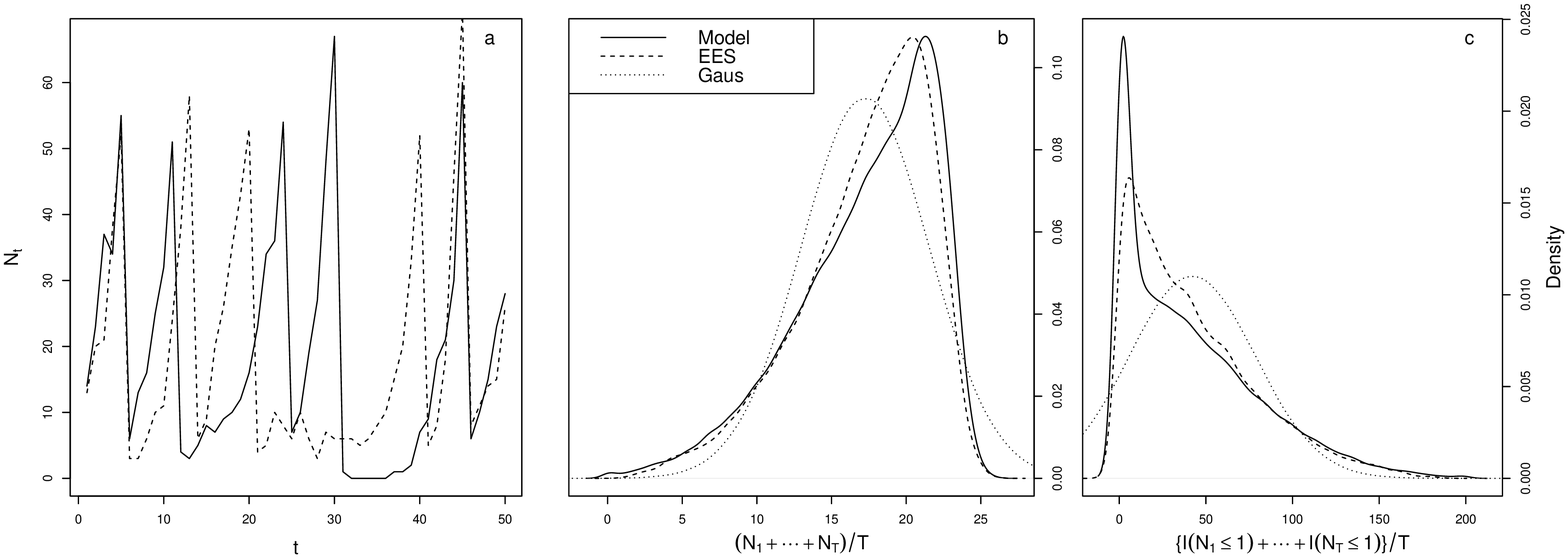}
\caption{a: Two sample paths, simulated from model (\ref{eq:crashModel}). b and c: Marginal distributions of the population mean and of the number of periods $t$ during which $N_t \leq 1$, when $T=250$. Both distributions are highly skewed, and the EES density achieves a much better fit than a Gaussian approximation.}
\label{fig: crash_dist}
\end{figure}

The development of approximate inferential approaches, such as SL and ABC, has been driven by the increasing availability of computational resources and by the  challenges to model-based inference emerging in computational biology, ecology and genetics. These methods address the issue that, for many scientif{\text{}}ically motivated models, the likelihood function is intractable; it may be too expensive to evaluate, unknown, or too time-consuming to derive analytically. Furthermore, even when sophisticated integration approaches, such as particle f{\text{}}ilters \citep{doucet2000sequential}, could provide consistent likelihood estimators, using approximate methods might still be preferable in practice, because of speed, automation and robustness to implementation details. Particle f{\text{}}ilters often rely on the specif{\text{}}ic structure of the chosen model, hence their implementation may need substantial changes if a different model is considered. In contrast, SL and ABC treat the model as a black box, thus allowing practitioners to rapidly explore a variety of models. 

The performance of SL, ABC and particle filters has been compared in detail by \cite{fasiolo2016comparison} and \cite{everitt2015bayesian}, respectively in the context of parameter estimation for non-linear state space models and of Bayesian model selection for Markov random field models. It is not the purpose of this paper to provide another extensive comparison. Instead, we focus on SL, and we start from the observation that the above-mentioned properties of this method are not without cost. In fact, although the Central Limit Theorem assures asymptotic normality of many classes of statistics, improving the quality of the normal approximation is not easy in a multivariate setting. Finding a suitable normalizing transformation is particularly challenging in the context of parameter estimation, because such transformation would need to ensure normality across the parameter space. This motivates the main contribution of this work: we relax the multivariate normality assumption, while maintaining the ease-of-use and scalability of SL. We achieve this by proposing a f{\text{}}lexible density estimator, namely the EES approximation. In addition to illustrating empirically that, when the distribution of the summary statistics is far from normal, the EES-based version of SL leads to more accurate parameter estimates than its Gaussian alternative, we prove that maximizing the synthetic likelihood produces consistent parameter estimators, under either the EES or the Gaussian density estimator. 

The paper is organized as follows. We introduce the empirical saddlepoint approximation in Section \ref{sec:oldEmpSad} and we propose its extended version in Section \ref{sec:exSad}. In Section \ref{sec:withinSL} we clarify how the new density estimator can be used within the context of SL and we prove the consistency of the resulting parameter estimator. In Sections \ref{sec:crashModel} and \ref{sec:shifted} we illustrate the method on model (\ref{eq:crashModel}) and on another simple example, designed to show the potential limitations of ABC and of the Gaussian version of SL, and in Section \ref{sec:formind_model} we apply the method to inference for a complex individual-based forest model, for which statistical inference is challenging without the use of summary statistics, while the model is sufficiently computationally costly that extensive method tuning is impractical. 

\section{Saddlepoint approximations} \label{sec:oldEmpSad}

Recall that we are interested in using saddlepoint methods to closely approximate the statistics-based likelihood, $p(\bm s^0| \bm \theta)$. However, the following discussion is valid beyond the context of SL, hence we temporarily suppress all dependencies on $\bm \theta$. We restore them in Section \ref{sec:withinSL}, which describes how the proposed density estimator can be used within SL.

We were led to saddlepoint approximations, among other multivariate density estimators, by the following considerations. While saddlepoint approximations are derived from asymptotic expansions, they are often very accurate even in small samples and, in contrast to Edgeworth approximations, they are strictly positive and do not show polynomial-like waves in the tails. In addition, their empirical version provides a close approximation to the density of widely used statistics, such as $M$- \citep{ronchetti1994empirical} and $L$-estimators \citep{easton1986general}.
 
Saddlepoint expansions were introduced into the statistical literature by \cite{daniels1954saddlepoint} and can be used to approximate the density function of a random variable, starting from its moment or cumulant generating function. When $\bm S$ is a continuous $d$-dimensional random vector, its probability density function, $p(\bm s)$, is associated with the moment generating function
\begin{equation*}
M( \bm \lambda) = \mathbb{E}\big (e^{ \bm \lambda^T  \bm S}\big ) = \int_{- \infty}^{+ \infty} e^{ \bm \lambda^T   \bm s}\, p( \bm s) \, d  \bm s,
\end{equation*}
while the cumulant generating function is def{\text{}}ined as $K( \bm \lambda) = \log{M( \bm \lambda)} $. We indicate its gradient and Hessian with $K'(\bm \lambda)$ and $K''(\bm \lambda)$. In the following we assume that $M( \bm \lambda)$ exists for $ \bm \lambda \in I$, where $I$ is a non-vanishing subset of $\mathbb{R}^d$ containing the origin. If $ \bm S$ is a discrete random vector, the generating functions are obtained by substituting the integrals with summations over the support of $ \bm S$.

Saddlepoint approximations rely on the one-to-one correspondence between the cumulant generating function and the probability density function of $ \bm S$. For a continuous $\bm S$, the saddlepoint density is
\begin{equation*}
\hat{p}( \bm s) = \frac{1}{ (2 \pi)^{\frac{d}{2}} \, \text{det}\{K''(\hat{ \bm \lambda})\}^{\frac{1}{2}} } 
e^{K(\hat{ \bm \lambda}) - \hat{ \bm \lambda}^T  \bm s} \; , 
\end{equation*}
where $\hat{ \bm \lambda}$ is such that
\begin{equation} \label{eq:sadeq}
K'(\hat{ \bm \lambda}) =  \bm s.
\end{equation}
Condition (\ref{eq:sadeq}) is often called the saddlepoint equation. The saddlepoint density is def{\text{}}ined only on the interior, $J_{V_{\bm s}}$, of the support, $V_{\bm s}$, of the original density, $p( \bm s)$.
Another important property of $\hat{p}( \bm s)$ is that it is generally improper. A proper density can be obtained through normalization
$$
\bar{p}( \bm s) = \frac{\hat{p}( \bm s)}{\int_{J_{V_{\bm s}}}\hat{p} ( \bm s) d    \bm s}.
$$
For a discrete $ \bm S$ analogous results hold and $\bar{p}( \bm s)$ should be interpreted as an approximation to $\text{pr}( \bm S =  \bm s)$.  For an introduction to saddlepoint approximations, see \cite{butler2007saddlepoint}.

\subsection{Empirical Saddlepoint approximation} \label{sec:empSad}

Suppose that the analytic form of $K( \bm \lambda)$ is unknown, as it generally is for simulation-based methods such as SL. If we can simulate from $p( \bm s)$, then it is possible to estimate $K( \bm \lambda)$ using the estimator proposed by \cite{davison1988saddlepoint}
\begin{equation} \label{eq:emp_cgf}
\hat{K}_m( \bm \lambda) = \log{\hat{M}_m( \bm \lambda)} = \log{ \bigg ( \frac{1}{m}\sum_{i=1}^{m}e^{ \bm \lambda^T   \bm s_i} \bigg )},
\end{equation}
where $m$ is the number of simulations used. Derivative estimates of $\hat{K}_m( \bm \lambda)$ are
$$
\hat{K}_m'( \bm \lambda) = \frac{\sum_{i=1}^{m}e^{ \bm \lambda^T   \bm s_i} \bm s_i}{\sum_{i=1}^{m}e^{ \bm \lambda^T  \bm s_i}},\;\;\;\;\;\;\;\; \hat{K}_m''( \bm \lambda) = \frac{\sum_{i=1}^{m}e^{ \bm \lambda^T  \bm s_i} \bm s_i  \bm s_i^T}{\sum_{i=1}^{m}e^{ \bm \lambda^T  \bm s_i}} - \hat{K}_m'( \bm \lambda) \hat{K}_m'( \bm \lambda)^T.
$$
These can be used to obtain an empirical saddlepoint approximation
\begin{equation} \label{eq:empSadEstim}
\hat{p}_m( \bm s) = \frac{1}{ (2 \pi)^{\frac{d}{2}} \, \text{det}\{\hat{K}_m''(\hat{ \bm \lambda}_m)\}^{\frac{1}{2}} } 
e^{\hat{K}_m(\hat{ \bm \lambda}_m) - \hat{ \bm \lambda}_m^T  \bm s},
\end{equation} 
where $\hat{ \bm \lambda}_m$ is the solution of
\begin{equation} \label{eq:empSadEq}
\hat{K}'_m(\hat{ \bm \lambda}_m) = \bm s.
\end{equation}
Notice that $\hat{K}_m'(\bm \lambda)$ is a convex combination of the simulated vectors $ \bm s_i$, hence (\ref{eq:empSadEq}) has no solution if $ \bm s$ falls outside the convex hull of the $ \bm s_i$s. This limitation is addressed in Section \ref{sec:exSad}.

\cite{feuerverger1989empirical} provides asymptotic results regarding how well $\hat{p}_m( \bm s)$ approximates $\hat{p}( \bm s)$ in a univariate setting. In the Supplementary Material we show how these carry over to the current multivariate setting. In particular, $\hat{p}_m( \bm s)$ converges to $\hat{p}( \bm s)$ at parametric rate $O(m^{-1/2})$ for $ \bm \lambda \in I/2$, where $I/2$ is the subset of $I$ such that $ \bm \lambda \in I/2$ if $2  \bm \lambda \in I$, while the convergence is slower outside this region. Regardless of the distribution of $ \bm S$, $ \bm s =  \bm \mu = \mathbb{E}( \bm S)$ corresponds to $ \bm \lambda =  \bm{0} \in I/2$, hence it might be advantageous to think of $\hat{K}_m'(I/2)$ as a region approximately centred around $ \bm \mu$. In Section \ref{sec:exSad} we build upon this interpretation.

\section{Extended Empirical Saddlepoint approximation} \label{sec:exSad}

The aim of this work is to use the f{\text{}}lexibility of the empirical saddlepoint approximation to estimate densities for which the normal approximation is poor. The asymptotic results of \cite{feuerverger1989empirical} suggest that the saddlepoint approximation should perform reasonably well in the central part of the distribution, while its accuracy decreases in the tails. More importantly, as stated in Section \ref{sec:empSad}, the empirical saddlepoint equation (\ref{eq:empSadEq}) has a solution only if $ \bm s$ lies inside the convex hull of the simulated data, so the resulting empirical saddlepoint density is not def{\text{}}ined outside this subset of $\mathbb{R}^d$.
This is problematic in the context of SL because, whether we wish to estimate the unknown parameters by Maximum Likelihood or Markov chain Monte Carlo (MCMC), we cannot generally expect ${\bm s}^0$ to fall inside the convex hull of the simulated statistics in early iterations. In addition, if the model of interest is unable to generate summary statistics that are close to the observed ones, its inadequacy should ideally be quantif{\text{}}ied by a low, rather than an undef{\text{}}ined, value of the synthetic likelihood. Hence, we need a remedy that allows us to solve (\ref{eq:empSadEq}) for any $\bm s = {\bm s}^0$.

To motivate our solution, notice that solving (\ref{eq:empSadEq}) is equivalent to minimizing
\begin{equation*}
\{ \hat{K}_m( \bm \lambda) -  \bm \lambda^T  \bm s \}^2,
\end{equation*}
which would be guaranteed to have a unique minimum, if strong convexity held. That is, if
\begin{equation} \label{eq:convex_cond}
\exists \; \epsilon \in \mathbb{R}^+ \;\; \text{such that} \;\;  \bm z^T \hat{K}_m''( \bm \lambda)  \bm z > \epsilon || \bm z||^2, \;\;\;\; \forall \;\;  \bm \lambda,  \bm z \in \mathbb{R}^d \;\; \text{such that} \;\; || \bm z|| > 0, 
\end{equation}
then (\ref{eq:empSadEq}) could be solved for any $ \bm s$. Unfortunately, the following proposition states that this in not the case.

\begin{Proposition} 
$\hat{K}_m( \bm \lambda)$ is strictly, but not strongly, convex. 
\end{Proposition}
\begin{proof}
See Appendix \ref{app:convex_proof}. 
\end{proof}
However, the fact that $\hat{K}( \bm \lambda)$ is strictly convex assures that tilting this estimator with a strongly convex function will produce a modif{\text{}}ied estimator that is strongly convex itself, so that (\ref{eq:empSadEq}) could be solved for any $ \bm s$. Therefore, we propose to use a modif{\text{}}ied estimator
\begin{equation} \label{eq:modCGF}
\hat{K}_m( \bm \lambda, \gamma, \bm s) = g( \bm s, \gamma) \hat{K}_m( \bm \lambda) + \{ 1 - g( \bm s, \gamma) \} \hat{G}_m( \bm \lambda),
\end{equation}
where $\hat{G}_m( \bm \lambda)$ is a strongly convex function, while $g( \bm s, \gamma)$ is a function of $ \bm s$, parametrized by $\gamma$, which determines the mix between the two functions. Furthermore, we require
\begin{equation} \label{eq:gamma_req}
g( \bm s, \gamma) : \mathbb{R}^d \to [0, 1], \;\;\;\;\;\;\;
\lim_{|| \bm s -  \hat{\bm \mu}|| \to \infty} g( \bm s, \gamma) = 0.
\end{equation} 
A natural choice for $\hat{G}_m( \bm \lambda)$ is the following parametric estimator of $K( \bm \lambda)$
\begin{equation} \label{eq:normalCGF}
\hat{G}_m( \bm \lambda) = { \bm \lambda^T}\hat{ \bm \mu} + 
\frac{1}{2} { \bm \lambda^T}\hat{ \bm \Sigma}{ \bm \lambda},
\end{equation}
which is unbiased and consistent for multivariate normal random variables. 
%
This leads to 
$$
\hat{K}_m( \bm \lambda, \gamma, \bm s) = { \bm \lambda^T}\hat{ \bm \mu} + 
\frac{\eta}{2} { \bm \lambda^T}\hat{ \bm \Sigma}{ \bm \lambda} + g( \bm s, \gamma) \bigg \{ \frac{1}{3!} \sum_{i = 1}^d\sum_{j = 1}^d\sum_{k = 1}^d  \frac{\partial^3 \hat{K}_m}{\partial \lambda_i\partial\lambda_j\partial\lambda_k}\bigg|_{\bm \lambda = \bm 0}\lambda_i\lambda_j\lambda_k + \cdots \bigg \},$$
where $\eta = 1 - g( \bm s, \gamma)/(m-1)$ appears because here $\hat{ \bm \Sigma}$ is the standard unbiased covariance estimator, while $\hat{K}''_m(\bm \lambda = \bm 0) = m \hat{\bm \Sigma} / (m-1)$. Similarly, evaluating higher derivatives of $\hat{K}_m(\bm \lambda)$ at $\bm \lambda = \bm 0$ produces consistent, but biased, estimators of the corresponding cumulants. Unbiased cumulant estimators are the $k$-statistics \citep{mccullagh1987tensor}. 
%

Our solution is related to that of \cite{wang1992general}, who modif{\text{}}ied the truncated estimator of \cite{easton1986general}, and to the proposal of \cite{bartolucci2007penalized}, in the context of Empirical Likelihood \citep{owen2001empirical}.
%
%
We refer to the density obtained by using estimator (\ref{eq:modCGF}) within (\ref{eq:empSadEstim}) as the Extended Empirical Saddlepoint approximation (EES). In Section \ref{sec: trans_fun} we propose a particular form for $g( \bm s, \gamma)$.

\subsection{Choosing and tuning the mixture function $g( \bm s, \gamma)$} \label{sec: trans_fun}

In Appendix \ref{app:MSEcalc} we derive the MSEs of estimators (\ref{eq:emp_cgf}) and (\ref{eq:normalCGF}), under normality of $ \bm S$. We then base our choice of $g( \bm s, \gamma)$ on the relative MSE performance of the two estimators. In particular, we choose 
\begin{equation} \label{eq:mix_fun_choice}
\begin{array} {lcl} 
g( \bm s, \gamma) & = & \Bigg [ \cfrac{{( \bm s - \hat{ \bm \mu})^T}  \hat{\bm \Sigma}^{-1} ( \bm s - \hat{ \bm \mu}) \big \{ 1 + \frac{1}{2} ( \bm s - \hat{ \bm \mu})^T  \hat{\bm \Sigma}^{-1} ( \bm s - \hat{ \bm \mu}) \big \}+ 1}{\text{exp} \big \{ ( \bm s - \hat{ \bm \mu}) ^ T  \hat{\bm \Sigma}^{-1} ( \bm s - \hat{ \bm \mu}) \big \}} \Bigg ]^{\gamma} \\
& \approx & \bigg[\cfrac{\text{MSE} \{ \hat{G}_m(\bm \lambda) \}+1}{\text{MSE} \{ \hat{K}_m(\bm \lambda) \}+1} \bigg]^{\gamma}, 
\end{array}
\end{equation}
where $\gamma > 0$ is a tuning parameter, which determines the rate at which $g( \bm s, \gamma)$ varies from $1$ to $0$, as the distance between $ \bm s$ and $\hat{ \bm \mu}$ increases. Apart from fulfi{\text{}}lling requirement (\ref{eq:gamma_req}), function (\ref{eq:mix_fun_choice}) has the desirable property of being invariant under linear transformations. More precisely, if $\bm z = \bm a + \bm B \bm s$ and $\bm Z_i = \bm a + \bm B \bm S_i$, for $i = 1, \dots, m$, then $g^{\bold z}(\bm z, \gamma) = g(\bm s, \gamma)$. Using this fact, it is simple to show that EES is equivariant under such transformations, that is $\log \hat{p}^{\bm z}_m(\bm z, \gamma) = \log \hat{p}_m(\bm s, \gamma) - \log\text{det}(\bm B)$. In practice, this allows us to normalize $\bm s$ and $\bm S_1, \dots, \bm S_m$ before fitting, which generally enhances numerical stability.


Our choice (\ref{eq:mix_fun_choice}) has two main shortcomings: it is based on a normality assumption for $ \bm S$ and, most importantly, it does not take the sample size $m$ into account. In regard to the first issue: using higher moments of the simulated statistics to determine (\ref{eq:mix_fun_choice}) might be attractive, but our experience suggests that this would result in very unstable estimates. The second problem can be addressed by appropriately selecting the tuning parameter $\gamma$. Its value is critical for the performance of our method, and at f{\text{}}irst sight it not clear on what principle this choice should be based. However, saddlepoint approximations are exact for Gaussian densities \citep{butler2007saddlepoint}, hence $\gamma$ is fundamentally a complexity-controlling parameter, which determines the balance between two density estimators: the empirical saddlepoint, which is characterized by higher f{\text{}}lexibility and variance, and the normal distribution, which generally has higher bias, but lower variance. Hence, we propose to select $\gamma$ by $k$-fold cross-validation, as detailed in the Algorithm \ref{alCross}.
\begin{algorithm}
\caption{Cross-validation with nested normalization}
\label{alCross}
\begin{algorithmic}[1]
\STATE Create a grid of $r$ possible values, $\gamma_1, \dots, \gamma_r$, for the tuning parameter.
\STATE Simulate $m$ random vectors $ \bm S_1, \dots,  \bm S_m$ from the true density $p( \bm s)$ and divide them into $k$ folds. For simplicity, assume that $m$ is a multiple of $k$. Indicate with $\bar{\bm S}_t$ the vectors in the $t$-th fold, and with $\bar{\bm S}_{-t}$ the remaining $r = m(1-1/k)$ vectors. Let $\hat{p}_{r}^{-t}( \bm s, \gamma)$ be the EES density based on the vectors in $\bar{\bm S}_{-t}$.
\STATE For $i = 1, \dots, r$
\begin{enumerate}[leftmargin=0.5cm, itemindent=0cm]
\item[] For $t = 1, \dots, k$
\begin{enumerate}[leftmargin=0.5cm, itemindent=0.5cm]
\item[$\cdot$] Estimate the normalizing constant of $\hat{p}_{r}^{-t}( \bm s, \gamma_i)$ by importance sampling, that is
$$
\hat{z}_r^{-t}(\gamma_i) = \frac{1}{l} \sum_{j = 1}^l \frac{\hat{p}_r^{-t}( \bm S_j^*, \gamma_i)}{q( \bm S_j^*)}, 
\;\;\;\;\;
\bm S_j^* \sim q( \bm s), \;\;\; \text{for} \;\; j = 1,\dots,l. 
$$
A reasonably eff{\text{}}icient importance density $q( \bm s)$ can be obtained by f{\text{}}itting a multivariate normal density to the vectors in $\bar{\bm S}_{-t}$. Notice that (\ref{eq:gamma_req}) and (\ref{eq:normalCGF}) assure the boundedness of the importance weights, under this choice of $q( \bm s)$.
\item[$\cdot$] Using the normalized EES density,  
$$
\bar{p}_r^{-t}( \bm s, \gamma_i) = 
\frac{\hat{p}_r^{-t}( \bm s, \gamma_i)}{\hat{z}_r^{-t}(\gamma_i)},
$$
evaluate the negative log-likelihood of the validation data $\bar{\bm S}_t$.

\end{enumerate}
\end{enumerate}

\STATE Select the value $\gamma_i$ that minimizes the negative validation log-likelihood, averaged across the $k$ folds.
\end{algorithmic}
\end{algorithm}

In the Supplementary Material we show that, as $m$ and $l \rightarrow \infty$, Algorithm \ref{alCross} consistently selects the value of $\gamma$ which minimizes the Kullback-Leibler divergence between $\bar{p}( \bm s, \gamma)$ and $p(\bm s)$. The Gaussian case is recovered as $\gamma \rightarrow \infty$.


\section{Use within Synthetic Likelihood} \label{sec:withinSL}


We now describe how the proposed density estimator can be used within the context of SL, hence we restore all dependencies on the model parameters, $\bm \theta$. To obtain an initial estimate, $\bm \theta_I$, of the unknown parameters it is reasonable to maximize the synthetic likelihood based on the Gaussian approximation, which is less computationally expensive. Then, $\gamma$ can be selected using Algorithm \ref{alCross}, with $p( \bm s) = p( \bm s|  \bm \theta_I)$. Given $\gamma$, pointwise estimates of the synthetic likelihood can be based on the new density estimator by using a procedure analogous to Algorithm \ref{alPoint}, which we describe in the Supplementary Material.

In terms of computational effort, if we assume that $m$, the number of summary statistics simulated from $p(\bm s|\bm \theta)$, is much larger than $d$, then the cost of evaluating the Gaussian synthetic likelihood is $O(m d^2)$, which is the cost of obtaining $\hat{\bm \Sigma}_{\bm \theta}$. Calculating $\hat{K}_m''(\bm \lambda)$ has the same complexity, but solving the empirical saddlepoint equation (\ref{eq:empSadEq}) numerically implies that $\hat{K}_m''(\bm \lambda)$ will be evaluated at several values of $\bm \lambda$. The proposal described in Section \ref{sec:exSad} assures that the underlying root f{\text{}}inding problem is strongly convex, hence few iterations of a Newton-Raphson algorithm are generally suff{\text{}}icient to solve (\ref{eq:empSadEq}) with high accuracy. The computational cost of a synthetic likelihood estimate is then O($lmd^2)$, if the normalizing constant is estimated using $l$ importance samples. In practice, we have not yet encountered an example where the normalizing constant strongly depended on $\bm \theta$. However, the normalizing constant often varies significantly with $\gamma$. Hence, in the examples presented in this paper, we estimate the normalizing constant when selecting $\gamma$ using 
Algorithm \ref{alCross}, but we use the unnormalized EES density during parameter estimation.  


Before testing the ESS-based version of SL on the examples, we now prove that, under the conditions to be specified shortly, maximizing the synthetic likelihood leads to consistent parameter estimators. Here we denote the Gaussian-based synthetic likelihood with $\hat{p}_{G}(\bm s^0|\bm \theta)$ and its EES-based version with $\hat{p}_{S}(\bm s^0|\bm \theta)$. We firstly consider the Gaussian case and we prove identifiability, which means that the (scaled) synthetic likelihood converges to a function which is maximized at the true parameter vector, ${\bm \theta}_0$. This is guaranteed under the following assumptions.

\begin{Assumption} \label{assConv}
The summary statistics depend on a set of underlying observations $\bm Y_{1}, \dots, \bm Y_{n}$, and have mean and covariance matrix
$$
\bm \mu_{\bm \theta}^{n}=\mathbb{E}(\bm S_n|\,\bm \theta),\;\;\;\;\bm \Sigma_{\bm \theta}^{n}=\mathbb{E}\big\{(\bm S_n-\bm \mu_{\bm \theta}^{n})(\bm S_n-\bm \mu_{\bm \theta}^{n})^{T}|\, \bm \theta \big\},
$$
where $\bm S_n=S(\bm Y_{1},\dots,\bm Y_{n})$. In addition there exists $\delta>0$ such that,  for any $\bm \theta$, we have
$$
\hat{\bm \mu}_{\bm \theta}^{n} \rightarrow \bm \mu_{\bm \theta}  \;\;\;\;\text{and\;\;}\;\; n^{\delta}\hat{\bm \Sigma}_{\bm \theta}^{n}\rightarrow\bm \Sigma_{\bm \theta},
$$
in probability, as $m$ and $n \rightarrow \infty$. 
\end{Assumption}

\begin{Assumption} \label{assMaxEig}
Let ${}_*{\psi}_{\bm \theta}$ and ${}^*{\psi}_{\bm \theta}$ be, respectively, the smallest and the largest eigenvalues of the asymptotic (scaled) covariance matrix, $\bm \Sigma_{\bm \theta}$. There exists two positive constants, ${}_*{\psi}$ and ${}^*{\psi}$, such that ${}_*{\psi}_{\bm \theta}>{}_*{\psi}$ and ${}^*{\psi}_{\bm \theta}<{}^*{\psi}$ for any $\bm \theta$. 
\end{Assumption}

\begin{Assumption} \label{assOne}
$\bm \mu_{\bm \theta} = \mu(\bm \theta)$ is one to one.
\end{Assumption}

\begin{Theorem} \label{theoGaus}
If assumptions \ref{assConv} to \ref{assOne} hold, as $m$ and $n \rightarrow \infty$  the scaled synthetic log-likelihood, $n^{-\delta}\log\hat{p}_{G}(\bm s^0|\bm \theta)$, is asymptotically proportional to a function, $f_{\bm \theta_0}(\bm \theta)$, which is maximal at $\bm \theta = \bm \theta_0$.  
\end{Theorem}
\begin{Proof}
See Appendix \ref{app:idenGaus}. 
\end{Proof}

Here assumption \ref{assConv} guarantees pointwise convergence to $f_{\bm \theta_0}(\bm \theta)$, while assumptions \ref{assMaxEig} and  \ref{assOne} assure identifiability. The fact that $f_{\bm \theta_0}(\bm \theta)$ is maximal at the true parameter  is not itself sufficient to assure weak consistency, which is instead guaranteed under the additional condition that the convergence of the Gaussian synthetic likelihood is uniform \citep{van2000asymptotic}. To assure this, we make the following assumptions.

\begin{Assumption} \label{assComp}
The parameter space, $\bm \Theta \subset \mathbb{R}^p$, is compact and convex.
\end{Assumption}

\begin{Assumption} \label{assDerivBound}
The derivatives of $\hat{\bm \mu}_{\bm \theta}^{n}$ and $n^{\delta} \hat{\bm \Sigma}_{\bm \theta}^{n}$ are continuous and dominated by  two $O_p(1)$ positive random sequences, $a_{n,m}$ and $b_{n,m}$. More precisely 
$$
\bigg | \bigg| \frac{\partial \hat{\bm \mu}_{\bm \theta}^{n}}{\partial \theta_k} \bigg | \bigg|_2 \leq a_{n,m}, \;\;\;\;\; \text{and} \;\;\;\;\; \bigg|\bigg| \frac{\partial n^{\delta} \hat{\bm \Sigma}_{\bm \theta}^{n}}{\partial \theta_k} \bigg|\bigg|_2 \leq b_{n,m},
$$ 
for $k = 1, \dots, q$ and for any $\bm \theta \in \bm \Theta$.
\end{Assumption}

\begin{Assumption} \label{assWellCond}
Let ${}_*\hat{\psi}_{\bm \theta}^n$ and ${}^*\hat{\psi}_{\bm \theta}^n$ be, respectively, the smallest and the largest eigenvalue of  
$n^{\delta}\hat{\bm \Sigma}_{\bm \theta}^{n}$ and assume that there  exist two $O_p(1)$ positive random sequences, $c_{n, m}$ and $u_{n, m}$, such that  
$$
\big({}_*\hat{\psi}_{\bm \theta}^n\big)^{-1} \leq c_{n,m}, \;\;\; \text{and} \;\;\; {}^*\hat{\psi}_{\bm \theta}^n \leq u_{n,m},
$$ 
for any $\bm \theta \in \bm \Theta$.
\end{Assumption}

\begin{Assumption} \label{assEquicont}
The derivatives of the asymptotic mean vector, $\bm \mu_{\bm \theta}$, and (scaled) covariance matrix, $\bm \Sigma_{\bm \theta}$, are bounded. In particular, there exist two positive constants, $M_{\bm \mu}$ and $M_{\bm \Sigma}$, such that 
$$
\bigg | \bigg| \frac{\partial \bm \mu_{\bm \theta}}{\partial \theta_k} \bigg | \bigg|_2 \leq M_{\bm \mu}, \;\;\;\;\; \text{and} \;\;\;\;\; \bigg|\bigg| \frac{\partial {\bm \Sigma}_{\bm \theta}}{\partial \theta_k} \bigg|\bigg|_2 \leq M_{\bm \Sigma},
$$ 
for any $\bm \theta \in \bm \Theta$.
\end{Assumption}

\cite{newey1991uniform} shows that pointwise convergence in probability, proved as part of theorem \ref{theoGaus}, implies uniform convergence as long as: assumption \ref{assComp} holds; the derivatives of $n^{-\delta}\log\hat{p}_{G}(\bm s^0|\bm \theta)$ are continuous and dominated by an $O_p(1)$ random sequence; $f_{\bm \theta_0}(\bm \theta)$ is equicontinuous. Hence, in the Gaussian case, proving consistency requires only assuring that the last two requirements are met. 

\begin{Theorem} \label{theoConsGaus}
Let $\hat{\bm \theta}$ be the maximizer of the Gaussian synthetic likelihood. If assumptions \ref{assConv} to \ref{assEquicont} hold, then $\hat{\bm \theta}$ converges in probability to ${\bm \theta}_0$, as $m$ and  $n \rightarrow \infty$.
\end{Theorem}
\begin{Proof}
See Appendix \ref{app:consGaus}.
\end{Proof}

We now consider the EES-based synthetic likelihood, and we focus on the un-normalized density estimator, which is cheaper to compute in practice. To prove identifiability we require the two following conditions to hold, in addition to assumptions \ref{assConv}, \ref{assMaxEig} and \ref{assOne}.

\begin{Assumption} \label{assMoment}
For every $n$, the moment generating function of $\bm S_n$ exists for $ \bm \lambda \in I$, where $I$ is a non-vanishing subset of ${\mathbb{R}}^d$ containing the origin.
\end{Assumption}

\begin{Assumption} \label{assGamma}
Let $\hat{\gamma}_{\bm \theta_I}^{n}$ be the chosen tuning parameter, corresponding to simulation effort $m$ and sample size $n$. As $m$ and $n \rightarrow \infty$, there exists a constant $c>0$ such that
$$\text{Prob}(\hat{\gamma}_{\bm \theta_I}^{n} < c) \rightarrow 0,
$$
for any initialization $\bm \theta_{I}$.
\end{Assumption}

\begin{Theorem} \label{theoSaddle}
If assumptions \ref{assConv}, \ref{assMaxEig}, \ref{assOne}, \ref{assMoment} and \ref{assGamma} hold, as $m$ and $n \rightarrow \infty$  the scaled synthetic log-likelihood, $n^{-\delta}\log\hat{p}_{S}(\bm s^0|\bm \theta)$, is asymptotically proportional to a function, $f_{\bm \theta_0}(\bm \theta)$, which is maximal at $\bm \theta_0.$ 
\end{Theorem}
\begin{Proof}
See Appendix \ref{app:consSaddle}. 
\end{Proof}

Notice that the asymptotic function, $f_{\bm \theta_0}(\bm \theta)$, mentioned in theorems \ref{theoGaus} and \ref{theoSaddle}, is the same under either density estimator. As in the Gaussian case, weak consistency is guaranteed under identifiability and uniform convergence to $f_{\bm \theta_0}(\bm \theta)$. Given assumption \ref{assComp}, and the fact that the equicontinuity of $f_{\bm \theta_0}(\bm \theta)$ has already been proven in the proof of theorem \ref{theoConsGaus}, uniform convergence is guaranteed as long as the derivatives of $n^{-\delta}\log\hat{p}_{S}(\bm s^0|\bm \theta)$ are continuous and dominated by an $O_p(1)$ sequence. In the Gaussian case this was assured under assumptions  on the derivatives of $\hat{\bm \mu}_{\bm \theta}^{n}$ and $n^{\delta} \hat{\bm \Sigma}_{\bm \theta}^{n}$, and on the eigenvalues of the latter. Given the complexity of the EES density, under this density estimator we prefer to impose conditions directly on $n^{-\delta}\log\hat{p}_{S}(\bm s^0|\bm \theta)$, rather than on $K(\bm \lambda)$, $K'(\bm \lambda)$ and $K''(\bm \lambda)$.

\begin{Assumption} \label{assDerivBoundSaddle}
The derivatives of the synthetic log-likelihood based on the EES density are continuous and dominated by an $O_p(1)$ random sequence, $v_{n, m}$, that is 
$$
\bigg| \frac{\partial n^{-\delta}\log\hat{p}_{S}(\bm s^0|\bm \theta)}{\partial \theta_k} \bigg | \leq v_{n, m},
$$ 
for $k = 1, \dots, q$ and for any $\bm \theta \in \bm \Theta$.
\end{Assumption}

\begin{Theorem} \label{theoConsSaddle}
Let $\hat{\bm \theta}$ be the maximizer of the EES-based synthetic likelihood. If assumptions \ref{assConv}, \ref{assMaxEig}, \ref{assOne}, \ref{assComp}, \ref{assEquicont}, \ref{assMoment}, \ref{assGamma} and \ref{assDerivBoundSaddle} hold, then $\hat{\bm \theta}$ converges weakly to ${\bm \theta}_0$, as $m$ and $n \rightarrow \infty$.
\end{Theorem}
\begin{Proof}
Theorem \ref{theoSaddle} assures pointwise convergence and identifiability under assumptions \ref{assConv}, \ref{assMaxEig}, \ref{assOne}, \ref{assMoment} and \ref{assGamma}. Pointwise converge, together with assumptions \ref{assComp}, \ref{assEquicont} and \ref{assDerivBoundSaddle} guarantee uniform convergence in probability \citep{newey1991uniform}. Uniform convergence and identifiability are sufficient conditions for weak consistency \citep{van2000asymptotic}.
\end{Proof}

\section{Simple recruitment, boom and bust model} \label{sec:crashModel}

Figure \ref{fig: crash_dist}a shows two trajectories simulated from model (\ref{eq:crashModel}), using parameters $r=0.4$, $\kappa=50$, $\alpha=0.09$ and $\beta=0.05$. To compare ABC with the Gaussian and EES-based version of SL, we simulate 100 pseudo-observed datasets of length $T=300$ using the above parameters. Given that we are not interested in estimating $N_0$, we discards the first 50 times steps to lose any transient behaviour from the system. We use the remaining 250 steps of each trajectory to compute the following summary statistics: mean and smallest population, number of times the population consists of one or less individuals, number of population peaks and square-root of the minimal time gap between two consecutive peaks (a peak is occurring at time $t$ if $N_{t+1}-N_{t} \leq 30$). Under ABC, we obtain MAP estimates of the model parameters using the approach of \cite{rubio2013simple}. This consists of sampling the approximate posterior, and then maximizing a kernel density estimate of it. We perform the sampling step using the MCMC approach of \cite{marjoram2003markov}, and we maximize the approximate posterior using the mean shift algorithm \citep{fukunaga1975estimation}. The ABC tolerance was chosen using the approach of \cite{wegmann2009efficient}, which will be described in Section \ref{sec:shifted}, using $10^6$ simulations and target acceptance rate $10^{-3}$. We use the uniform priors $r\in(0, 1)$, $\kappa\in(10, 80)$, $\alpha \in (0, 1)$ and $\beta\in(0, 1)$, so that MLE and MAP estimates are equivalent. For SL we use $m=5\times10^{3}$ and estimate $\gamma=4 \times 10^{-4}$ using Algorithm \ref{alCross}, with $l=10^4$. While the Supplementary Material gives further details about the simulation setting, it is important to point out here that we use the same number of simulations from the model under all methods.

\begin{figure}
\centering
\includegraphics[scale = 0.43]{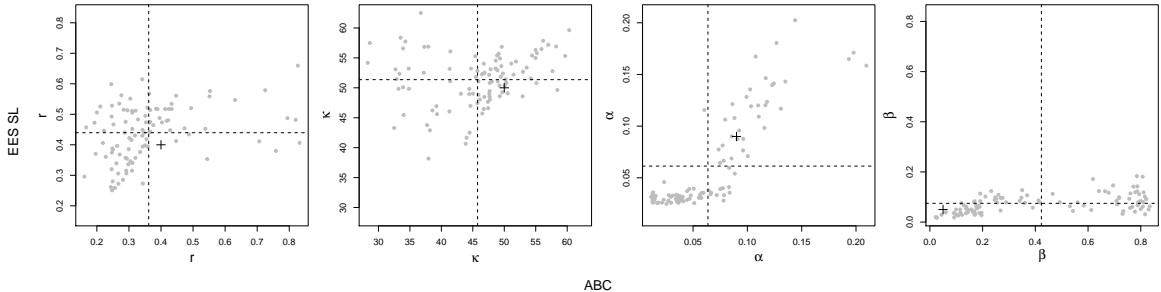}
\caption{MLE estimates obtained using EES SL versus MAP estimates produced by ABC, under model (\ref{eq:crashModel}). The dashed lines and the black crosses indicate, respectively, the mean estimates under each method and the true parameter values.} 
\label{fig: EES_ABC}
\end{figure}

Figure 2 compares the estimates obtained using ABC and EES-based SL, while Table 1 reports the true parameters, together with the means and RMSEs for all three methods.  ABC struggles to identify the arrival rate, $\beta$, and it often grossly overestimates $r$ and underestimates
$\kappa$, so that its RMSE performances is substantially worse overall than that of either SL method. Comparing the synthetic likelihood methods to each other, EES-SL has lower RMSE than Gaussian SL for all parameters. The mean for $\beta$ is also substantially closer to the true parameter for EES-SL, as expected given the shape of the distribution of the most relevant statistic, shown in Figure 1c. This very simple example clearly illustrates that EES-SL offers non-negligible benefits when important statistics are highly non-Gaussian.

\begin{table}  
\begin{center}
\begin{tabular}{cccccc} 
 \\
\hline
Param. & Truth & ABC & Gaus. SL & EES SL & Scale  \\
\hline 
    $r$ & $4$ &  3.6(1.5) &  4.5(1.1) &  4.4(\underline{0.97}) & $10^{-1}$   \\ 
    $\kappa$ & $50$ &  45.7(8.5)  &  51.7(4.8) &  51.4(\underline{4.5}) &  1  \\ 
    $\alpha$ & $9$ &  6.3(\underline{5}) &  6.2(5.7) &  6.1(5.4) & $10^{-2}$   \\ 
    $\beta$ & $5$ &  42.4(46.7) &  10.5(7.3) &  7.5(\underline{4.4}) & $10^{-2}$   \\
\hline
\end{tabular} 
\end{center}
\caption{True parameters, means and RMSEs (in parentheses) of the estimates using ABC and the two version of SL, for model (\ref{eq:crashModel}). For each parameter, the lowest RMSE is underlined.}
\label{tab:MSE_crash}
\end{table}  


\section{Multivariate shifted exponential distribution} \label{sec:shifted}

Here we consider a toy example, whose purpose is illustrating how the performance of the Gaussian and EES versions of SL compares with that of tolerance-based ABC algorithms, as the dimensionality of the model increases. In particular, let $ \bm S$ be a $d$-dimensional random vector, where each marginal follows a shifted exponential distribution
\begin{equation} \label{eq:shiftModel}
S_k \sim \theta_k + \text{Exp}(\beta), \;\;\; \text{for} \; k = 1, \dots, d.
\end{equation}
The plot in Figure \ref{fig: cv_plot}a contains the results of a 10-fold cross-validation run, obtained using $d = 10$, $l = 10^3$, $m = 10^4$, $\beta = 0.5$ and $\theta_1 = \cdots = \theta_d = 0$. The cross-validation curve is minimized by $\gamma = 5 \times 10^{-3}$, and the plot in Figure \ref{fig: cv_plot}b shows the true and approximate marginal densities of one component $S_k$. The EES approximation to the marginal, obtained by marginalizing the $d$-dimensional f{\text{}}it, is clearly more accurate than a normal density. 

\begin{figure}
\centering
\includegraphics[scale = 0.5]{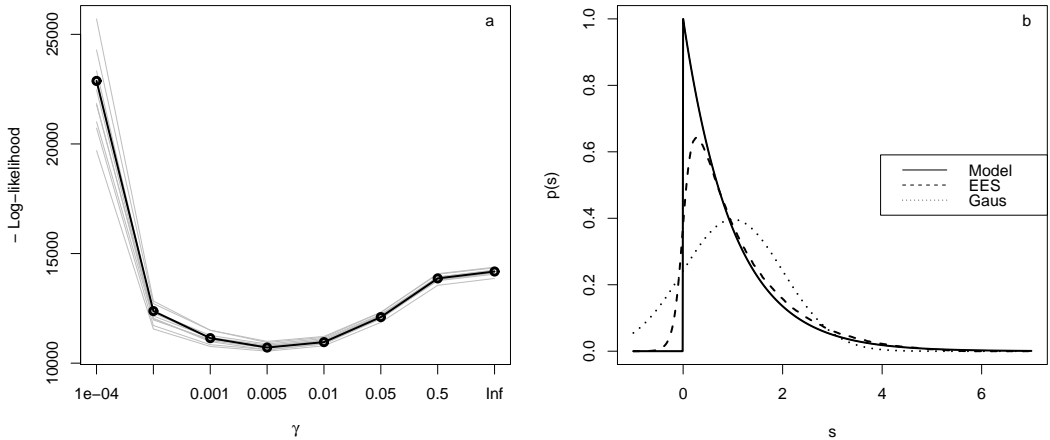}
\caption{Shifted exponential model. a: Curves from 10-fold cross-validation, the black line is their average. b: True $\text{Exp}(\beta)$ density (black), EES (dashed) and normal (dotted) approximation.}
\label{fig: cv_plot}
\end{figure}


To demonstrate the usefulness of EES in the context of SL, we use it to estimate the shifts $\theta_1, \dots, \theta_d$, all of which are equal to $0$. In particular, we simulate a single vector of observed statistics, $\bm s^0$, from (\ref{eq:shiftModel}), and we maximize the resulting synthetic likelihood, using either a Gaussian density or EES. Given that SL and ABC algorithms are often motivated as approximations to full likelihood inference, all the MSEs reported in the remainder of this section quantify deviations from the full MLE, which is $\bm s^0$, not from the true parameters. Hence, the bias of Gaussian SL estimates is $1/\beta$. By averaging the squared errors across the 10 dimensions, we obtain MSEs equal to 3.8 and 0.56, using the normal and the EES approximation respectively. In an analogous 20-dimensional run, using $m = 5\times 10 ^4$, the MSE was reduced from 4.1 to 1.26. P-values from t-test for differences in log-absolute errors were around $10^{-6}$ in both runs. 

It is possible to derive analytically how a tolerance-based ABC approximation would perform under this model. The details are reported in the Supplementary Material. Assume that the likelihood $p(\bm s^0|\bm \theta)$ is approximated by $p(||\bm s^0 - \bm s||_{\infty} < \epsilon|\bm \theta)$, where $\epsilon > 0$ is the tolerance. Given that we are interested in deviations from the full MLE, we can impose $\bm s^0 = \bm 0$ without loss of generality. If we use independent uniform priors on $[\psi, 0]$ for each parameter, where $\psi < -\epsilon$, the posterior mode is at $\theta_k = -\epsilon$, for $k=1,\dots,d$. Hence, the MSE corresponding to the MAP estimate is equal to $\epsilon^2$. This implies that, to achieve MSEs equal to those of EES, $\epsilon$ would need to be set to $\sqrt{0.56}$ and $\sqrt{1.26}$, respectively in the 10 and 20-dimensional setting. The corresponding acceptance probabilities, obtained by simulating statistic vectors using parameters $\bm \theta$ fixed to the MAP, are of order $10^{-3}$ and $10^{-4}$. In 40 dimensions, obtaining an MSE equal to $2$ would lead to an acceptance ratio at the MAP of order $10^{-5}$. Notice that these are upper bounds, because the acceptance probability is maximal at the MAP.

This analysis suggests that, in order to match the MSE achieved by SL, the computational budget of an ABC algorithm would need to be increased rapidly as the number of dimensions grows. Further, the relation between $\epsilon$ and the MSE is generally not known in practice. A popular approach (see e.g. \cite{wegmann2009efficient}) is to simulate a large number of parameter vectors from the prior, then simulate a statistics vector from the model using each of them and select $\epsilon$ so that a small percentage of these are accepted. To quantify how the computational cost of this tuning phase depends on the prior, we reverse this process and assume that the values of $\epsilon$ are as given above. In 10 dimensions, $\psi$ would need be in $[-2.1, -\epsilon]$, in order to achieve an acceptance probability of order $10^{-4}$, while in 20 dimension, $\psi \in [-1.9, -\epsilon]$ leads an acceptance probability of order $10^{-5}$. Hence, to obtain only just tolerable acceptance rates during the tuning phase, very accurate prior information must be available, especially in high dimensions. In an applied setting, prior information is often rather vague, hence $\epsilon$ needs to be tuned using more sophisticated approaches, such as the sequential algorithm of \cite{toni2009approximate}. While such methods can alleviate the effort needed to select $\epsilon$, performing extensive ABC tuning runs is still onerous when working with computationally intensive models, such as the one described in the next section.

\section{Formind forest model} \label{sec:formind_model}

\subsection{The model}

To test our proposal in a realistic setting, we consider Formind, an individual-based model describing the main natural processes driving forests dynamics. Here we describe its basic features, while we refer to \cite{dislich2009simulating} and to \cite{fischer2016lessons} for detailed descriptions of the model and of the scientific questions it can be used to address.

The model describes the growth and population dynamics of tree individuals in a simulation area that is divided in $20\times20$m patches, with individual trees being assigned explicitly to one patch. Tree species with similar characteristics are grouped into Plant Functional Types (PFTs). A constant input of seeds deposits on average $s_j$ seeds of the $j$-th PFT per hectare per year. The main factor determining both seed establishment and growth is the light climate in the patch. For example, pioneer types will establish only in patches relatively free of overshadowing trees, while late successional trees are able to grow below a dense canopy. Trees are subject to a baseline mortality rates $m_{j}$, which is specif{\text{}}ic to each PFT.

%

In the context of Formind, the need for approximate simulation-based methods comes from the complexity of the model. Indeed, Formind was developed with a focus on ecological plausibility, rather than statistical tractability, and most of its submodels describe highly non-linear biological processes, containing one or more sources of randomness. Most importantly, the raw output of Formind is the collection of all the characteristics of individual trees in the simulations area, which obviously do not correspond to individuals present in the actual survey data. Hence, it is necessary to work with summary statistics.

Formind is computationally intensive even when few PFTs are included and, given initial conditions and parameters, the simulated forest needs to be run for several hundred years, before the distribution of the summary statistics reaches equilibrium. This means that, from a practical point of view, it is critical to avoid lengthy tuning runs, such as those needed to select the tolerance $\epsilon$ in ABC methods.

\subsection{Simulation Results}  \label{sec:formind_simul}

We consider two PFTs, pioneer and late successional, and we reduce the model output to 6 summary statistics. In particular, to verify whether then new density estimator can deal with large departures from normality, we used the following transformed statistics   
$$
S_{jk} = \alpha_{jk}^{\frac{ C_{jk} - \psi_{jk} }{\sigma_{jk}}},
\;\; \text{for} \; j \in \{1, 2\}, \; k \in \{1, 2, 3\},    
$$
where $C_{jk}$ is the number of trees of the $j$-th PFT falling in the $k$-th diameter class, while $\alpha_{jk}$, $\psi_{jk}$, and $\sigma_{jk}$ are constants, whose values are reported in the Supplementary Material. The diameter categories used for each PFT correspond to trees with small, medium or large diameters.

We simulated 24 datasets from the model and estimated the baseline mortality rates and seed input intensities of the two PFTs by maximizing the synthetic likelihood, using both the normal and the EES approximations. In both cases we used $m = 10^4$ simulated summary statistics and, under EES, $\gamma$ was f{\text{}}ixed to $5.5 \times 10^{-3}$, chosen using Algorithm \ref{alCross} with $l = 10^3$. Table \ref{tab:MSE_eesa} reports the true parameters, together with the means and RMSEs of the estimates, from the normal or the EES approximations. See the Supplementary Material for more details about the optimization setting.
\begin{table}  
\begin{center}
\begin{tabular}{cccccc} 
 \\
\hline
Param. & Truth & Gaus. SL & EES SL & Scale & P-value \\
\hline 
    $\mu_{pio}$ & $5$ & 4.7 (1.4) & 5.4 (0.7) & $10^{-2}$ & 0.002 \\ 
    $\mu_{suc}$ & $5$ & 9.3 (6.5) & 6.1 (1.6) & $10^{-3}$ & 0.003 \\ 
    $s_{pio}$ & $80$ & 108.4 (41.1) & 91.6 (26.2) & 1 & 0.07 \\ 
    $s_{suc}$ & $20$ & 31.6 (15.7) & 23.2 (4.7) & 1 & 0.003 \\
\hline
\end{tabular} 
\end{center}
\caption{Formind model: true parameters, means and RMSEs (in parentheses) of the estimates using the normal and the EES estimators. P-values for differences in log-absolute errors have been calculated using t-tests.}
\label{tab:MSE_eesa}
\end{table} 

Using the EES, rather than the normal approximation, leads to lower MSEs for all model parameters. The plots in Figure \ref{fig:forest_hist} compare the marginal distributions of the summary statistics, simulated from the model using the true parameter values, with those obtained by simulating random vectors from EES, f{\text{}}itted to the simulated statistics using the same values of $\gamma$ and $m$ used during the optimization. EES gives a good f{\text{}}it to the marginal distributions of the summary statistics, all of which are far from normal.

\begin{figure} 
\centering
\includegraphics[scale=0.6]{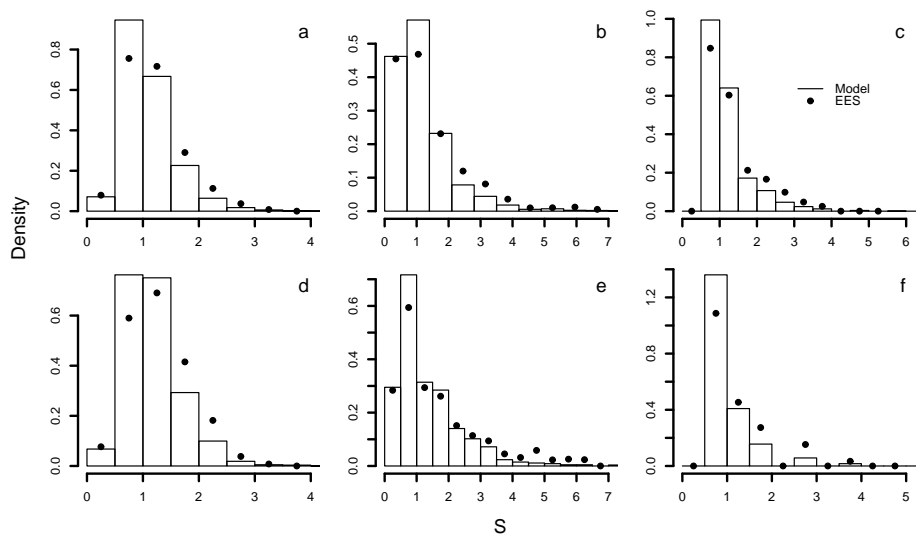}
\caption{Marginal distributions of summary statistics corresponding to small, medium and large pioneers (a, b, c) and successionals (d, e, f), in the Formind model.}
\label{fig:forest_hist}
\end{figure} 

\section{Conclusions} \label{sec:conclusions}

In this work we have relaxed the normality assumption, which characterized the original formulation of SL, by proposing a novel, more flexible, density estimator. As the examples show, EES scales well with the number of summary statistics used and it is able to model densities for which a normal approximation is clearly inadequate. This in turn can lead to better accuracy in parameter estimation.

Importantly, using EES rather than a Gaussian density, does not add much to the tuning requirements of SL. In fact, the only parameter of EES, $\gamma$, can be selected using standard statistical tools, such as cross-validation. In the context of SL, and of approximate methods in general, having little tuning requirements is an important feature, since it allows practitioners to focus on identifying informative summary statistics, rather than on other aspects of the inferential procedure. 

We have shown that, if fairly general conditions on the distribution of the summary statistics and on the underlying model hold, maximizing the synthetic likelihood function leads to consistent parameter estimators, under either EES or a Gaussian density estimator. Given the generality of the conditions assumed here, we have treated EES as a non-parametric density estimator, as suggested by \cite{feuerverger1989empirical}. However, the use of empirical saddlepoint approximations has previously been considered for particular classes of statistics, such as M-estimators \citep{monti1993relationship,ronchetti1994empirical} and L-statistics \citep{easton1986general}. Hence, it would be interesting to verify whether making additional assumptions on the summary statistics would allow us to  assess the asymptotic efficiency of the parameter estimates produced by the EES-based version of SL.


From a practical point of view, the computational eff{\text{}}iciency of SL is of critical importance.  \cite{gutmann2015bayesian}, \cite{wilkinson2014accelerating} and \cite{meeds2014gps} proposed using Gaussian Processes to increase the computational eff{\text{}}iciency of SL and ABC methods. The f{\text{}}irst two proposals, being based on pointwise likelihood estimates, could be used in conjunction with EES. \cite{meeds2014gps} model only the f{\text{}}irst two moments of the simulated statistics, hence it is not clear whether their approach could be modif{\text{}}ied to take higher moments into account, as the new density estimator does.

\section*{Acknowledgement}
MF and SNW have been partly funded by the EPSRC grants EP/I000917/1 and EP/K005251/1.

\bibliographystyle{chicago}
\bibliography{biblio.bib}

\begin{thebibliography}{}

\bibitem[\protect\citeauthoryear{Bartolucci}{Bartolucci}{2007}]{bartolucci2007penalized}
Bartolucci, F. (2007).
\newblock A penalized version of the empirical likelihood ratio for the
  population mean.
\newblock {\em Statistics \& probability letters\/}~{\em 77\/}(1), 104--110.

\bibitem[\protect\citeauthoryear{Beaumont, Zhang, and Balding}{Beaumont
  et~al.}{2002}]{beaumont2002approximate}
Beaumont, M.~A., W.~Zhang, and D.~J. Balding (2002).
\newblock Approximate bayesian computation in population genetics.
\newblock {\em Genetics\/}~{\em 162\/}(4), 2025--2035.

\bibitem[\protect\citeauthoryear{Blum, Nunes, Prangle, and Sisson}{Blum
  et~al.}{2013}]{blum2013comparative}
Blum, M., M.~Nunes, D.~Prangle, and S.~Sisson (2013).
\newblock A comparative review of dimension reduction methods in approximate
  bayesian computation.
\newblock {\em Statistical Science\/}~{\em 28\/}(2), 189--208.

\bibitem[\protect\citeauthoryear{Blum}{Blum}{2010}]{blum2010approximate}
Blum, M.~G. (2010).
\newblock Approximate bayesian computation: a nonparametric perspective.
\newblock {\em Journal of the American Statistical Association\/}~{\em
  105\/}(491).

\bibitem[\protect\citeauthoryear{Butler}{Butler}{2007}]{butler2007saddlepoint}
Butler, R.~W. (2007).
\newblock {\em Saddlepoint approximations with applications}.
\newblock Cambridge University Press.

\bibitem[\protect\citeauthoryear{Cherubini, Luciano, and Vecchiato}{Cherubini
  et~al.}{2004}]{cherubini2004copula}
Cherubini, U., E.~Luciano, and W.~Vecchiato (2004).
\newblock {\em Copula methods in f{\text{}}inance}.
\newblock John Wiley \& Sons.

\bibitem[\protect\citeauthoryear{Daniels}{Daniels}{1954}]{daniels1954saddlepoint}
Daniels, H.~E. (1954).
\newblock Saddlepoint approximations in statistics.
\newblock {\em The Annals of Mathematical Statistics\/}~{\em 25\/}(4),
  631--650.

\bibitem[\protect\citeauthoryear{Davison and Hinkley}{Davison and
  Hinkley}{1988}]{davison1988saddlepoint}
Davison, A.~C. and D.~V. Hinkley (1988).
\newblock Saddlepoint approximations in resampling methods.
\newblock {\em Biometrika\/}~{\em 75\/}(3), 417--431.

\bibitem[\protect\citeauthoryear{Dislich, G{\"u}nter, Homeier, Schr{\"o}der,
  and Huth}{Dislich et~al.}{2009}]{dislich2009simulating}
Dislich, C., S.~G{\"u}nter, J.~Homeier, B.~Schr{\"o}der, and A.~Huth (2009).
\newblock Simulating forest dynamics of a tropical montane forest in south
  ecuador.
\newblock {\em Erdkunde\/}~{\em 63\/}(4), 347--364.

\bibitem[\protect\citeauthoryear{Doucet, Godsill, and Andrieu}{Doucet
  et~al.}{2000}]{doucet2000sequential}
Doucet, A., S.~Godsill, and C.~Andrieu (2000).
\newblock On sequential monte carlo sampling methods for bayesian
  f{\text{}}iltering.
\newblock {\em Statistics and computing\/}~{\em 10\/}(3), 197--208.

\bibitem[\protect\citeauthoryear{Doucet, Jacob, and Rubenthaler}{Doucet
  et~al.}{2013}]{doucet2013derivative}
Doucet, A., P.~E. Jacob, and S.~Rubenthaler (2013).
\newblock Derivative-free estimation of the score vector and observed
  information matrix with application to state-space models.
\newblock {\em arXiv preprint arXiv:1304.5768\/}.

\bibitem[\protect\citeauthoryear{Easton and Ronchetti}{Easton and
  Ronchetti}{1986}]{easton1986general}
Easton, G.~S. and E.~Ronchetti (1986).
\newblock General saddlepoint approximations with applications to $\text{L}$
  statistics.
\newblock {\em Journal of the American Statistical Association\/}~{\em
  81\/}(394), 420--430.

\bibitem[\protect\citeauthoryear{Everitt, Johansen, Rowing, and
  Evdemon-Hogan}{Everitt et~al.}{2015}]{everitt2015bayesian}
Everitt, R.~G., A.~M. Johansen, E.~Rowing, and M.~Evdemon-Hogan (2015).
\newblock Bayesian model comparison with intractable likelihoods.
\newblock {\em arXiv preprint arXiv:1504.00298\/}.

\bibitem[\protect\citeauthoryear{Fasiolo, Pya, Wood, et~al.}{Fasiolo
  et~al.}{2016}]{fasiolo2016comparison}
Fasiolo, M., N.~Pya, S.~N. Wood, et~al. (2016).
\newblock A comparison of inferential methods for highly nonlinear state space
  models in ecology and epidemiology.
\newblock {\em Statistical Science\/}~{\em 31\/}(1), 96--118.

\bibitem[\protect\citeauthoryear{Feuerverger}{Feuerverger}{1989}]{feuerverger1989empirical}
Feuerverger, A. (1989).
\newblock On the empirical saddlepoint approximation.
\newblock {\em Biometrika\/}~{\em 76\/}(3), 457--464.

\bibitem[\protect\citeauthoryear{Fischer, Bohn, de~Paula, Dislich, Groeneveld,
  Guti{\'e}rrez, Kazmierczak, Knapp, Lehmann, Paulick, et~al.}{Fischer
  et~al.}{2016}]{fischer2016lessons}
Fischer, R., F.~Bohn, M.~D. de~Paula, C.~Dislich, J.~Groeneveld, A.~G.
  Guti{\'e}rrez, M.~Kazmierczak, N.~Knapp, S.~Lehmann, S.~Paulick, et~al.
  (2016).
\newblock Lessons learned from applying a forest gap model to understand
  ecosystem and carbon dynamics of complex tropical forests.
\newblock {\em Ecological Modelling\/}.

\bibitem[\protect\citeauthoryear{Fukunaga and Hostetler}{Fukunaga and
  Hostetler}{1975}]{fukunaga1975estimation}
Fukunaga, K. and L.~Hostetler (1975).
\newblock The estimation of the gradient of a density function, with
  applications in pattern recognition.
\newblock {\em IEEE Transactions on information theory\/}~{\em 21\/}(1),
  32--40.

\bibitem[\protect\citeauthoryear{Gutmann and Corander}{Gutmann and
  Corander}{2016}]{gutmann2015bayesian}
Gutmann, M.~U. and J.~Corander (2016).
\newblock Bayesian optimization for likelihood-free inference of
  simulator-based statistical models.
\newblock {\em The Journal of Machine Learning Research\/}~{\em 17\/}(1),
  4256--4302.

\bibitem[\protect\citeauthoryear{Hartig, Dislich, Wiegand, and Huth}{Hartig
  et~al.}{2014}]{hartig2014technical}
Hartig, F., C.~Dislich, T.~Wiegand, and A.~Huth (2014).
\newblock Technical note: Approximate bayesian parameterization of a
  process-based tropical forest model.
\newblock {\em Biogeosciences\/}~{\em 11}, 1261--1272.

\bibitem[\protect\citeauthoryear{Ionides, Bhadra, Atchad{\'e}, and
  King}{Ionides et~al.}{2011}]{ionides2011iterated}
Ionides, E.~L., A.~Bhadra, Y.~Atchad{\'e}, and A.~King (2011).
\newblock Iterated f{\text{}}iltering.
\newblock {\em The Annals of Statistics\/}~{\em 39\/}(3), 1776--1802.

\bibitem[\protect\citeauthoryear{Ionides, Bret{\'o}, and King}{Ionides
  et~al.}{2006}]{ionides2006inference}
Ionides, E.~L., C.~Bret{\'o}, and A.~A. King (2006).
\newblock Inference for nonlinear dynamical systems.
\newblock {\em Proceedings of the National Academy of Sciences\/}~{\em
  103\/}(49), 18438--18443.

\bibitem[\protect\citeauthoryear{Joe}{Joe}{2006}]{joe2006generating}
Joe, H. (2006).
\newblock Generating random correlation matrices based on partial correlations.
\newblock {\em Journal of Multivariate Analysis\/}~{\em 97\/}(10), 2177--2189.

\bibitem[\protect\citeauthoryear{Marjoram, Molitor, Plagnol, and
  Tavar{\'e}}{Marjoram et~al.}{2003}]{marjoram2003markov}
Marjoram, P., J.~Molitor, V.~Plagnol, and S.~Tavar{\'e} (2003).
\newblock Markov chain monte carlo without likelihoods.
\newblock {\em Proceedings of the National Academy of Sciences\/}~{\em
  100\/}(26), 15324--15328.

\bibitem[\protect\citeauthoryear{McCullagh}{McCullagh}{1987}]{mccullagh1987tensor}
McCullagh, P. (1987).
\newblock {\em Tensor methods in statistics}, Volume 161.
\newblock Chapman and Hall London.

\bibitem[\protect\citeauthoryear{Meeds and Welling}{Meeds and
  Welling}{2014}]{meeds2014gps}
Meeds, E. and M.~Welling (2014).
\newblock Gps-abc: Gaussian process surrogate approximate bayesian computation.
\newblock {\em arXiv preprint arXiv:1401.2838\/}.

\bibitem[\protect\citeauthoryear{Monti and Ronchetti}{Monti and
  Ronchetti}{1993}]{monti1993relationship}
Monti, A.~C. and E.~Ronchetti (1993).
\newblock On the relationship between empirical likelihood and empirical
  saddlepoint approximation for multivariate m-estimators.
\newblock {\em Biometrika\/}~{\em 80\/}(2), 329--338.

\bibitem[\protect\citeauthoryear{Newey}{Newey}{1991}]{newey1991uniform}
Newey, W.~K. (1991).
\newblock Uniform convergence in probability and stochastic equicontinuity.
\newblock {\em Econometrica: Journal of the Econometric Society\/}, 1161--1167.

\bibitem[\protect\citeauthoryear{Owen}{Owen}{2001}]{owen2001empirical}
Owen, A.~B. (2001).
\newblock {\em Empirical likelihood}.
\newblock CRC press.

\bibitem[\protect\citeauthoryear{Rao}{Rao}{2009}]{rao2009linear}
Rao, C.~R. (2009).
\newblock {\em Linear statistical inference and its applications}, Volume~22.
\newblock John Wiley \& Sons.

\bibitem[\protect\citeauthoryear{Rencher and Christensen}{Rencher and
  Christensen}{2012}]{rencher2012methods}
Rencher, A.~C. and W.~F. Christensen (2012).
\newblock {\em Methods of multivariate analysis}, Volume 709.
\newblock John Wiley \& Sons.

\bibitem[\protect\citeauthoryear{Roberts and Varberg}{Roberts and
  Varberg}{1973}]{robertsconvex}
Roberts, A.~W. and D.~E. Varberg (1973).
\newblock Convex functions.

\bibitem[\protect\citeauthoryear{Ronchetti and Welsh}{Ronchetti and
  Welsh}{1994}]{ronchetti1994empirical}
Ronchetti, E. and A.~H. Welsh (1994).
\newblock Empirical saddlepoint approximations for multivariate m-estimators.
\newblock {\em Journal of the Royal Statistical Society. Series B
  (Methodological)\/}~{\em 52\/}(2), 313--326.

\bibitem[\protect\citeauthoryear{Rubio, Johansen, et~al.}{Rubio
  et~al.}{2013}]{rubio2013simple}
Rubio, F.~J., A.~M. Johansen, et~al. (2013).
\newblock A simple approach to maximum intractable likelihood estimation.
\newblock {\em Electronic Journal of Statistics\/}~{\em 7}, 1632--1654.

\bibitem[\protect\citeauthoryear{Silverman}{Silverman}{1986}]{silverman1986density}
Silverman, B.~W. (1986).
\newblock {\em Density estimation for statistics and data analysis}, Volume~26.
\newblock CRC press.

\bibitem[\protect\citeauthoryear{Toni, Welch, Strelkowa, Ipsen, and
  Stumpf}{Toni et~al.}{2009}]{toni2009approximate}
Toni, T., D.~Welch, N.~Strelkowa, A.~Ipsen, and M.~P. Stumpf (2009).
\newblock Approximate bayesian computation scheme for parameter inference and
  model selection in dynamical systems.
\newblock {\em Journal of the Royal Society Interface\/}~{\em 6\/}(31),
  187--202.

\bibitem[\protect\citeauthoryear{Van~der Vaart}{Van~der
  Vaart}{2000}]{van2000asymptotic}
Van~der Vaart, A.~W. (2000).
\newblock {\em Asymptotic statistics}, Volume~3.
\newblock Cambridge university press.

\bibitem[\protect\citeauthoryear{Wang}{Wang}{1992}]{wang1992general}
Wang, S. (1992).
\newblock General saddlepoint approximations in the bootstrap.
\newblock {\em Statistics \& probability letters\/}~{\em 13\/}(1), 61--66.

\bibitem[\protect\citeauthoryear{Wegmann, Leuenberger, and Excoffier}{Wegmann
  et~al.}{2009}]{wegmann2009efficient}
Wegmann, D., C.~Leuenberger, and L.~Excoffier (2009).
\newblock Efficient approximate bayesian computation coupled with markov chain
  monte carlo without likelihood.
\newblock {\em Genetics\/}~{\em 182\/}(4), 1207--1218.

\bibitem[\protect\citeauthoryear{Wilkinson}{Wilkinson}{2014}]{wilkinson2014accelerating}
Wilkinson, R. (2014).
\newblock Accelerating abc methods using gaussian processes.
\newblock In {\em AISTATS}, pp.\  1015--1023.

\bibitem[\protect\citeauthoryear{Wood}{Wood}{2010}]{wood2010}
Wood, S.~N. (2010).
\newblock Statistical inference for noisy nonlinear ecological dynamic systems.
\newblock {\em Nature\/}~{\em 466\/}(7310), 1102--1104.

\bibitem[\protect\citeauthoryear{Yan et~al.}{Yan et~al.}{2007}]{yan2007enjoy}
Yan, J. et~al. (2007).
\newblock Enjoy the joy of copulas: with a package copula.
\newblock {\em Journal of Statistical Software\/}~{\em 21\/}(4), 1--21.

\end{thebibliography}

\begin{appendices}

\section{Proof of Proposition 1} \label{app:convex_proof}

Def{\text{}}ine
\begin{equation} \label{eq:pd_cond}
w_i = \frac{e^{ \bm \lambda^T  \bm s_i}}{\sum_{i=1}^me^{ \bm \lambda^T  \bm s_i}}, \;\;\;\;
\bar{ \bm s} = \hat{K}'( \bm \lambda) = \frac{\sum_{i = 1}^m w_i  \bm s_i}{\sum_{i = 1}^m w_i}, \;\;\;\; \text{i} = 1, \dots, m,
\end{equation}
and notice that $\hat{K}''( \bm \lambda)$ is positive semi-def{\text{}}inite
\begin{equation*}
\begin{array} {lcl} 
\bm z^T \hat{K}''( \bm \lambda)  \bm z &=&  \bm z^T \sum_{i = 1}^m w_i ( \bm s_i - \bar{ \bm s})( \bm s_i - \bar{ \bm s})^T  \bm z = \sum_{i = 1}^m w_i  \bm z^T ( \bm s_i - \bar{ \bm s})( \bm s_i - \bar{ \bm s})^T  \bm z \\ &=& \sum_{i = 1}^m w_i \big \{  \bm z^T ( \bm s_i - \bar{ \bm s}) \big \}^2 \geq 0, 
\end{array}
\end{equation*}
for all $\bm z \in \mathbb{R}^d$ such that $|| \bm z|| > 0$. In addition, def{\text{}}ine $ \bm q_i =  \bm s_i - \bar{ \bm s}$ and assume that
\begin{equation} \label{eq:rank_ass}
r = \text{rank} \, [ \bm q_1, \dots,  \bm q_m] = d.
\end{equation}
Then $\hat{K}''( \bm \lambda)$ is positive def{\text{}}inite and $\hat{K}( \bm \lambda)$ is strictly convex. In fact, suppose that there exists a non-zero vector $ \bm z$ such that $ \bm z^T \hat{K}''( \bm \lambda)  \bm z = 0$, which implies $\bm z^T \bm q_i = 0$ for $i = 1, \dots, m$. Given that $\bm z$ can be expressed as a linear combination of $\bm q_1, \dots, \bm q_m$, this would imply that
$$
 \bm z^T  \bm z =  (b_1  \bm q_1 + \dots + b_m  \bm q_m )^T  \bm z = 0,
$$
which contradicts the fact that $ \bm z$ is a non-zero vector. Now, def{\text{}}ine
$$
J \subset \big \{1, \dots, m \big \} \;\; \text{such that} \;\; 
 \bm \lambda^T  \bm s_{i} = \alpha > 0 \;\; \text{for all} \;\; i \in J, \;\;\;\;\; 
 \bm \lambda^T  \bm s_{i} < \alpha \;\; \text{for all} \;\; i \notin J,
$$
examination of (\ref{eq:pd_cond}) shows that
\begin{align*} 
\lim_{c \to \infty} w_i &= \frac{\lim_{c \to \infty} e^{c ( \bm \lambda^T  \bm s_i -  \bm \lambda^T  \bm s_j)}}{\lim_{c \to \infty} \sum_{k=1}^me^{c( \bm \lambda^T  \bm s_k -  \bm \lambda^T  \bm s_j)}} = \frac{0}{\text{Card}(J)} = 0, \;\;\; \text{for all} \;\; i, \, j \;\; \text{such that} \;\; j \in J, \; i \notin J, \\
\lim_{c \to \infty} w_i &= \frac{\lim_{c \to \infty} e^{c ( \bm \lambda^T  \bm s_i -  \bm \lambda^T  \bm s_j)}}{\lim_{c \to \infty} \sum_{k=1}^me^{c( \bm \lambda^T  \bm s_k -  \bm \lambda^T  \bm s_j)}} = \frac{1}{\text{Card}(J)}, \;\;\; \text{for all} \;\; i, \, j \;\; \text{such that} \;\; i,\, j \in J.
\end{align*}
Hence
$$
\lim_{c \to \infty} \bar{ \bm s} =  \lim_{c \to \infty} \hat{K}'(c  \bm \lambda) = \lim_{c \to \infty} \sum_{i = 1}^m w_i  \bm s_i = \frac{1}{\text{Card}(J)} \sum_{i \in J}^m  \bm s_i,
$$
and
$$
\lim_{c \to \infty}  \bm \lambda^T  \bm q_i = \lim_{c \to \infty}  \bm \lambda^T (  \bm s_i - \bar{ \bm s}) =  \bm \lambda^T \bigg \{  \bm s_i - \frac{1}{\text{Card}(J)} \sum_{i \in J}  \bm s_i \bigg \} = \alpha - \alpha = 0, \;\;\; \text{for all} \;\; i \in J.
$$
Finally, we choose $ \bm z =  \bm \lambda$ and obtain
$$
\lim_{c \to \infty}  \bm \lambda^T \hat{K}''(c  \bm \lambda)  \bm \lambda =   \sum_{i = 1}^m \lim_{c \to \infty} w_i \lim_{c \to \infty} \big (  \bm \lambda^T  \bm q_i \big )^2 = \frac{1}{\text{Card}(J)} \sum_{i \in J} \lim_{c \to \infty} \big (  \bm \lambda^T  \bm q_i \big )^2 = 0,
$$
which implies that $\hat{K}( \bm \lambda)$ is not strongly convex.

\section{Mean squared errors of the CGF estimators} \label{app:MSEcalc}

Firstly notice that, irrespective of the distribution of $ \bm S$, $\hat{M}( \bm \lambda)$ is unbiased. If $ \bm S$ is normally distributed, $e^{ \bm \lambda^T  \bm S}$ follows a log-normal distribution and
$$
M( \bm \lambda) = e^{ \bm \mu + \frac{1}{2} \bm \lambda ^ T  \bm \Sigma  \bm \lambda},
\;\;\;\;
\text{var} \big \{ \hat{M}( \bm \lambda)\big \} = \frac{1}{m} \big ( e^{ \bm \lambda ^ T  \bm \Sigma  \bm \lambda}  - 1 \big ) e^{2  \bm \mu +  \bm \lambda ^ T  \bm \Sigma  \bm \lambda},
$$
with the saddlepoint equation (\ref{eq:sadeq}) being solved by
\begin{equation} \label{eq:norm_sad_sol}
\hat{ \bm \lambda} = { \bm \Sigma}^{-1}( \bm s -  \bm \mu).
\end{equation}

In order to approximate the MSE of (\ref{eq:emp_cgf}) as a function of $ \bm \lambda$, we f{\text{}}irstly approximate its expected value by Taylor expansion around $M( \bm \lambda)$
\begin{equation*}
\begin{array} {lcl} 
\mathbb{E} \big \{\hat{K}( \bm \lambda) \big \}  & = &  \mathbb{E} \bigg [ \log{M( \bm \lambda)} + \cfrac{1}{M( \bm \lambda)}\big \{ \hat{M}( \bm \lambda) - M( \bm \lambda) \big \} - \cfrac{1}{2M( \bm \lambda)^2}\big \{ \hat{M}( \bm \lambda) - M( \bm \lambda) \big \}^2 + \cdots \bigg ] \\ & = &  \log{M( \bm \lambda)} - \cfrac{1}{2M( \bm \lambda)^2} \text{var}\big \{\hat{M}( \bm \lambda) \big \} + O(m^{-2}).
\end{array}
\end{equation*}
 Similarly we have that
\begin{equation*}
\begin{array} {lcl} 
\mathbb{E} \big \{ \hat{K}( \bm \lambda)^2 \big \}  & = &   \mathbb{E} \bigg [ \big \{ \log{M( \bm \lambda)} \big \}^2 + \cfrac{2 \log \{ M( \bm \lambda) \}}{M( \bm \lambda)}\big \{ \hat{M}( \bm \lambda) - M( \bm \lambda) \big \} \\& + & \bigg \{ \cfrac{1}{M( \bm \lambda)^2} - \cfrac{\log{M( \bm \lambda)}}{M( \bm \lambda)^2} \bigg \} \big \{ \hat{M}( \bm \lambda) - M( \bm \lambda) \big \} ^2 + \cdots \bigg ] \\ & = &   \big \{ \log{M( \bm \lambda)} \big \}^2 + \bigg\{\cfrac{1}{M( \bm \lambda)^2} -  \cfrac{\log{M( \bm \lambda)}}{M( \bm \lambda)^2} \bigg \}\text{var}\big \{\hat{M}( \bm \lambda) \big \} + O(m^{-2}),
\end{array}
\end{equation*}
hence
\begin{equation*}
\begin{array} {lcl} 
\text{var}\{\hat{K}( \bm \lambda)\}  & = &  \mathbb{E} \big \{\hat{K}( \bm \lambda)^2 \big \} - \mathbb{E} \big \{\hat{K}( \bm \lambda) \big \}^2 \\ & = &  \cfrac{1}{M( \bm \lambda)^2} \text{var}\big \{\hat{M}( \bm \lambda) \big \} - \cfrac{1}{4M( \bm \lambda)^4} \bigg [ \text{var} \big \{\hat{M}( \bm \lambda) \big \} \bigg ]^2  + O(m^{-2}).
\end{array}
\end{equation*}
Finally
\begin{equation} \label{eq: mseSad}
\begin{array} {lcl} 
\text{MSE}\{\hat{K}( \bm \lambda)\} & = &  \text{Bias}\{\hat{K}( \bm \lambda)\}^2 + \text{var}\{\hat{K}( \bm \lambda)\} 
\\& = & \cfrac{1}{M( \bm \lambda)^2} \text{var}\big \{\hat{M}( \bm \lambda) \big \} + O(m^{-2}) \\ & = & 
\cfrac{1}{m} \big ( e^{ \bm \lambda ^ T  \bm \Sigma  \bm \lambda}  - 1 \big ) + O(m^{-2}) 
\\ & = & \cfrac{1}{m} \big \{ e^{( \bm s -  \bm \mu) ^ T  \bm \Sigma ^{-1} ( \bm s -  \bm \mu)}  - 1 \big \} + O(m^{-2}),
\end{array}
\end{equation}
where the last equality holds due to (\ref{eq:norm_sad_sol}). The $O(m^{-2})$ term in (\ref{eq: mseSad}) derives from
\begin{equation*}
\begin{array} {lcl} 
\mathbb{E} \bigg [ \bigg \{ \hat{M}( \bm \lambda) - M( \bm \lambda) \bigg \}^3 \bigg ] & = & \mathbb{E} \bigg [ \bigg \{\frac{1}{m} \sum_{i=1}^m e^{ \bm \lambda^T  \bm S_i} - \mathbb{E}(e^{ \bm \lambda^T  \bm S}) \bigg \}^3 \bigg ] \\  & = & 
\cfrac{1}{m^3}   \sum_{i=1}^m \mathbb{E} \bigg [ \big \{ e^{ \bm \lambda^T  \bm S_i} - \mathbb{E}(e^{ \bm \lambda^T  \bm S}) \big \}^3 \bigg ] \\ & = & \cfrac{1}{m^2} \mu_3\big ({e^{ \bm \lambda^T  \bm S}}\big ),
\end{array}
\end{equation*}
where $\mu_3( X)$ is the third central moment of a random variable $ X$ and the second equality is justif{\text{}}ied by independence.

Estimator (\ref{eq:normalCGF}) is unbiased and consistent, if $ \bm S$ is normally distributed, hence
$$
\text{MSE}\{\hat{G}_m( \bm \lambda)\} = \text{var}\{\hat{G}_m( \bm \lambda)\} = { \bm \lambda^T}\text{var}(\hat{ \bm \mu}){ \bm \lambda} + \frac{1}{4} \text{var}\bigg ( { { \bm \lambda^T}\hat{ \bm \Sigma}{ \bm \lambda}} \bigg ), 
$$
due to the independence between $\hat{ \bm \mu}$ and $\hat{ \bm \Sigma}$ for normally distributed random variables (Basu's theorem). In addition, as $m$ goes to inf{\text{}}inity we have, from \cite{rencher2012methods}, that
$$
(m-1) \hat{ \bm \Sigma} = \sum_{i=1}^m ( \bm S_i - \hat{ \bm \mu})( \bm S_i - \hat{ \bm \mu})^T  \to  \bm W, \;\;\; \text{where} \;\;\; 
 \bm W \sim \text{Wishart}( \bm \Sigma, m-1),
$$
and from \cite{rao2009linear}
$$
 \bm \lambda^T  \bm W  \bm \lambda \sim \tau^2 Q, \;\; \text{where} \;\; 
\tau^2 =  \bm \lambda^T  \bm \Sigma  \bm \lambda \;\; \text{and} \;\; Q \sim \chi_{m-1}^2,
$$
hence, by using (\ref{eq:norm_sad_sol}), we obtain
\begin{equation*} \label{eq: mseNorm}
\begin{array} {lcl}
m\text{MSE}\{\hat{G}_m( \bm \lambda)\} & \to & {\hat{ \bm \lambda}^T} \bm \Sigma \hat{ \bm \lambda} + \cfrac{m}{2 (m-1)} ({\hat{ \bm \lambda}^T} \bm \Sigma \hat{ \bm \lambda})^2 \\ & \to & {\hat{ \bm \lambda}^T} \bm \Sigma \hat{ \bm \lambda} \big ( 1 + \cfrac{1}{2} {\hat{ \bm \lambda}^T} \bm \Sigma \hat{ \bm \lambda} \big ) \\ & = & {( \bm s -  \bm \mu)^T} \bm \Sigma^{-1} ( \bm s -  \bm \mu) \big \{ 1 + \cfrac{1}{2} ( \bm s -  \bm \mu)^T \bm \Sigma^{-1} ( \bm s -  \bm \mu) \big \}.
\end{array}
\end{equation*}

\section{Proof of Theorem \ref{theoGaus}} \label{app:idenGaus}

The Gaussian synthetic log-likelihood is proportional to
\begin{equation} \label{eq:synLogLikScaled}
n^{-\delta}\log{\hat{p}_{G}(\bm s^0|\bm \theta)}\propto-(\bm s^0-\hat{\bm \mu}_{\bm \theta}^{n})^{T}\,\big(n^{\delta}\hat{\bm \Sigma}_{\bm \theta}^{n}\big)^{-1}\,(\bm s^0-\hat{\bm \mu}_{\bm \theta}^{n})-n^{-\delta}\log\text{det}(\hat{\bm \Sigma}_{\bm \theta}^{n}).
\end{equation}
Assumption \ref{assConv} implies that
$$
n^{-\delta}\log\text{det}(\hat{\bm \Sigma}_{\bm \theta}^{n}n^{-\delta}n^{\delta})=n^{-\delta}\big\{\log\text{det}(\hat{\bm \Sigma}_{\bm \theta}^{n}n^{\delta})-d\delta\log n\big\}=O_p(n^{-\delta}) + O(n^{-\delta}\log n),
$$
so, as $n$ and $m \rightarrow \infty$, the r.h.s. of (\ref{eq:synLogLikScaled}) converges in probability to 
$$
f_{\bm \theta_0}(\bm \theta) = -(\bm \mu_{\bm \theta_{0}}-\bm \mu_{\bm \theta})^{T}\,\bm \Sigma_{\bm \theta}^{-1}\,(\bm \mu_{\bm \theta_{0}}-\bm \mu_{\bm \theta}),
$$
for any $\bm \theta$. Here $\bm \mu_{\bm \theta_{0}}$ is the asymptotic mean vector at true parameters $\bm \theta_{0}$. Then
$$
f_{\bm \theta_0}(\bm \theta) = -(\bm \mu_{\bm \theta_{0}}-\bm \mu_{\bm \theta})^{T}\,\bm {\bm U}_{\bm \theta} {\bm \Psi}_{\bm \theta} ^{-1} {\bm U}_{\bm \theta}^T\,(\bm \mu_{\bm \theta_{0}}-\bm \mu_{\bm \theta}) = - {\bm z}_{\bm \theta}^T {\bm \Psi}_{\bm \theta} ^{-1} {\bm z}_{\bm \theta},
$$
where $\bm {\bm U}_{\bm \theta} {\bm \Psi}_{\bm \theta} {\bm U}_{\bm \theta}^T$ is the eigen-decomposition of $\bm \Sigma_{\bm \theta}$, and we defined ${\bm z}_{\bm \theta}={\bm U}_{\bm \theta}^T\,(\bm \mu_{\bm \theta_{0}}-\bm \mu_{\bm \theta})$. Now, if $\bm \theta \neq \bm \theta_0$, then assumption \ref{assOne} assures that  $||\bm z_{\bm \theta}||_2=||\bm \mu_{\bm \theta_{0}}-\bm \mu_{\bm \theta}||_2>0$ which, together with assumption \ref{assMaxEig}, guarantees that
$$
f_{\bm \theta_0}(\bm \theta) = - \sum_{i=1}^d \frac{1}{(\bm \Psi_{\bm \theta})_{ii}} ({\bm z}_{\bm \theta})_i^2 \leq - \frac{1}{{}^*{\psi}} ||{\bm z}_{\bm \theta}||_2^2 < 0. 
$$
Given that $f_{\bm \theta_0}({\bm \theta_0})=0$, this function is maximized at $\bm \theta_0$, which implies identifiability under a Gaussian density estimator.
%

\section{Proof of Theorem \ref{theoConsGaus}} \label{app:consGaus}

Given assumption \ref{assDerivBound} and the fact that all the functions involved in $n^{-\delta}\log\hat{p}_{G}(\bm s^0|\bm \theta)$ are continuously differentiable, this function is continuously differentiable itself, due to the chain rule. We then have to show that its derivative is dominated by an $O_p(1)$ random sequence. Consider the partial derivative of the log-determinant w.r.t. the $k$-th parameter
$$
\bigg | \frac{\partial \log\text{det}(n^{\delta}\hat{\bm \Sigma}_{\bm \theta}^{n})}{\partial \theta_k} \bigg | = \bigg | \text{Tr} \bigg [ \big(n^{\delta}\hat{\bm \Sigma}_{\bm \theta}^{n}\big)^{-1}\frac{\partial n^{\delta} \hat{\bm \Sigma}_{\bm \theta}^{n}}{\partial \theta_k}   \bigg]\bigg | \leq  \text{Tr} \bigg [ \Big|\bm \hat{\bm U}_{\bm \theta}^n \Big| \Big| \big(\hat{\bm \Psi}^n_{\bm \theta}\big)^{-1} \Big| \Big| \bm (\hat{\bm U}_{\bm \theta}^n)^T \Big| \bigg | \frac{\partial n^{\delta} \hat{\bm \Sigma}_{\bm \theta}^{n}}{\partial \theta_k} \bigg |  \bigg],
$$
where $\bm \hat{\bm U}_{\bm \theta}^n \hat{\bm \Psi}^n_{\bm \theta} \bm (\hat{\bm U}_{\bm \theta}^n)^T$ is the eigen-decomposition of $n^{\delta}\hat{\bm \Sigma}_{\bm \theta}^{n}$. Then
\begin{equation}\label{eq:logDetDeriv}
\bigg | \frac{\partial \log\text{det}(n^{\delta}\hat{\bm \Sigma}_{\bm \theta}^{n})}{\partial \theta_k} \bigg | \leq \frac{1}{{}_*\hat{\psi}_{\bm \theta}^n} \text{Tr} \bigg [ \Big|\bm \hat{\bm U}_{\bm \theta}^n \Big|  \Big| \bm (\hat{\bm U}_{\bm \theta}^n)^T \Big| \bigg | \frac{\partial n^{\delta} \hat{\bm \Sigma}_{\bm \theta}^{n}}{\partial \theta_k} \bigg |  \bigg] \leq \frac{d}{{}_*\hat{\psi}_{\bm \theta}^n} \sum_{i=1}^d \sum_{j=1}^d \bigg | \frac{\partial n^{\delta} \hat{\bm \Sigma}_{\bm \theta}^{n}}{\partial \theta_k} \bigg |_{ij}.
\end{equation}
Now, consider the derivative of the inverse (scaled) covariance matrix
$$
\bigg|\bigg|\frac{\partial \big( n^{\delta} \hat{\bm \Sigma}_{\bm \theta}^{n} \big )^{-1}}{\partial \theta_k} \bigg|\bigg |_2 =  \bigg|\bigg|\big ( n^{\delta} \hat{\bm \Sigma}_{\bm \theta}^{n} \big )^{-1} \frac{\partial n^{\delta} \hat{\bm \Sigma}_{\bm \theta}^{n}}{\partial \theta_k} \big ( n^{\delta} \hat{\bm \Sigma}_{\bm \theta}^{n} \big )^{-1} \bigg|\bigg|_2 \leq \bigg|\bigg| \frac{\partial n^{\delta} \hat{\bm \Sigma}_{\bm \theta}^{n}}{\partial \theta_k} \bigg|\bigg|_2 \bigg|\bigg| \big ( n^{\delta} \hat{\bm \Sigma}_{\bm \theta}^{n} \big )^{-1} \bigg|\bigg|_2^2,
$$
but
$$
\Big|\Big| \big ( n^{\delta} \hat{\bm \Sigma}_{\bm \theta}^{n} \big )^{-1} \Big|\Big|_2 = \Big|\Big| \bm \hat{\bm U}_{\bm \theta}^n \Big|\Big|_2 \Big|\Big| \big(\hat{\bm \Psi}^n_{\bm \theta} \big)^{-1} \Big|\Big|_2 \Big|\Big| \bm (\hat{\bm U}_{\bm \theta}^n)^T \Big|\Big|_2 \leq \frac{1}{{}_*\hat{\psi}_{\bm \theta}^n}, 
$$
hence
\begin{equation} \label{eq:invDerivBound}
\bigg|\bigg|\frac{\partial \big( n^{\delta} \hat{\bm \Sigma}_{\bm \theta}^{n} \big )^{-1}}{\partial \theta_k} \bigg|\bigg |_2 \leq \big({}_*\hat{\psi}_{\bm \theta}^n\big)^{-2} \bigg|\bigg| \frac{\partial n^{\delta} \hat{\bm \Sigma}_{\bm \theta}^{n}}{\partial \theta_k} \bigg|\bigg|_2.
\end{equation}
Under assumptions \ref{assDerivBound} and \ref{assWellCond}, the r.h.s. of both (\ref{eq:logDetDeriv}) and (\ref{eq:invDerivBound}) are dominated. So if we consider
\begin{equation} \label{eq:derivGausLL}
\begin{array} {lcl}
\bigg| \cfrac{\partial n^{-\delta}\log\hat{p}_{G}(\bm s^0|\bm \theta)}{\partial \theta_k} \bigg | &\leq& \bigg|\bigg|\cfrac{\partial \hat{\bm \mu}_{\bm \theta}^{n}}{\partial \theta_k}\bigg|\bigg|_2\big|\big|n^{\delta} \hat{\bm \Sigma}_{\bm \theta}^{n}\big|\big|_2||\bm s_{0}-\hat{\bm \mu}^n_{\bm \theta}||_2 +  \frac{1}{2} \bigg|\bigg| \cfrac{\partial \big( n^{\delta} \hat{\bm \Sigma}_{\bm \theta}^{n} \big )^{-1}}{\partial \theta_k} \bigg|\bigg|_2 ||\bm s_{0}-\hat{\bm \mu}^n_{\bm \theta}||_2^2 \\  
&+& \frac{1}{2}\bigg|\cfrac{\partial \log\text{det}(n^{\delta}\hat{\bm \Sigma}_{\bm \theta}^{n})}{\partial \theta_k}\bigg|,
\end{array}
\end{equation}
it is clear that the r.h.s. of (\ref{eq:derivGausLL}) is dominated as well, provided that $||\bm s_0 - \hat{\bm \mu}^n_{\bm \theta}||_2$ is. Now
$$
||\bm s_0 - \hat{\bm \mu}^n_{\bm \theta}||_2^2 = ||\bm z_{\bm \theta_0}^n + {\bm \mu}^n_{\bm \theta_0} - \hat{\bm \mu}^n_{\bm \theta_0} + \hat{\bm \mu}^n_{\bm \theta_0} - \hat{\bm \mu}^n_{\bm \theta}||_2^2 \leq ||\bm z_{\bm \theta_0}^n||_2^2 + ||{\bm \mu}^n_{\bm \theta_0} - \hat{\bm \mu}^n_{\bm \theta_0}||_2^2 + ||\hat{\bm \mu}^n_{\bm \theta_0} - \hat{\bm \mu}^n_{\bm \theta}||_2^2,
$$
where $\bm z_{\bm \theta_0}^n = \bm s_0 - \bm \mu^n_{\bm \theta_0}$. But $||\hat{\bm \mu}^n_{\bm \theta_0} - \hat{\bm \mu}^n_{\bm \theta}||_2^2$ is dominated, because the derivatives of $\hat{\bm \mu}^n_{\bm \theta}$ are dominated by assumption \ref{assDerivBound} and the parameter space is compact by assumption \ref{assComp}. In addition, for any $\epsilon>0$, by Markov's inequality
\begin{equation} \label{eq:markov}
\text{Prob}\big(||\bm z_{\bm \theta_0}^n||_2^2>\epsilon\big) \leq \frac{\mathbb{E}\big(||\bm z_{\bm \theta_0}^n||_2^2\big)}{\epsilon} = n^{-\delta} \frac{\text{Tr}\big(n^{\delta} {\bm \Sigma}_{\bm \theta_0}^{n}\big)}{\epsilon} \leq  \frac{n^{-\delta}\, d\, {}^*{\psi}_{\bm \theta_0}^n }{\epsilon},
\end{equation}
where ${}^*{\psi}_{\bm \theta_0}^n$ is the largest eigenvalue of $n^{\delta} {\bm \Sigma}_{\bm \theta_0}^{n}$. The r.h.s. of (\ref{eq:markov}) is $O(n^{-\delta})$ by assumptions \ref{assConv} and \ref{assWellCond}, hence $||\bm z_{\bm \theta_0}^n||_2^2$ is $o_p(1)$.  An identical argument shows that $||{\bm \mu}^n_{\bm \theta_0} - \hat{\bm \mu}^n_{\bm \theta_0}||_2^2$ is $o_p(1)$. This proves that the derivatives of $n^{-\delta}\log\hat{p}_{G}(\bm s^0|\bm \theta)$ are dominated by an $O_p(1)$ sequence.

As explained in Section \ref{sec:withinSL}, this implies the uniform convergence of $n^{-\delta}\log\hat{p}_{G}(\bm s^0|\bm \theta)$ to $f_{\bm \theta_0}(\bm \theta)$, provided that $f_{\bm \theta_0}(\bm \theta)$ is equicontinuous. It is easy to show to that, if the derivatives of $f_{\bm \theta_0}(\bm \theta)$ are bounded, then equicontinuity follows. But, under assumptions \ref{assMaxEig} and \ref{assEquicont}, it is possible to bound the derivatives of $f_{\bm \theta_0}(\bm \theta)$, as we have just done for $n^{-\delta}\log\hat{p}_{G}(\bm s^0|\bm \theta)$. This assures uniform convergence which, together with identifiability, guarantees weak consistency \citep{van2000asymptotic}.

\section{Proof of Theorem \ref{theoSaddle}} \label{app:consSaddle}

Taylor expanding the un-normalized EES-based synthetic log-likelihood leads to
\begin{equation} \label{eq:limTail}
\log \hat{p}_S(\bm s^0|\bm \theta)=\log \hat{p}_G(\bm s^0|\bm \theta)+O\big\{ e^{-\hat{\gamma}_{\bm \theta}^{n}\,(\bm s^0-\hat{\bm \mu}_{\bm \theta}^{n})^{T}\,\big(\hat{\bm \Sigma}_{\bm \theta}^{n}\big)^{-1}\,(\bm s^0-\hat{\bm \mu}_{\bm \theta}^{n})}\big\},
\end{equation}
and, by multiplying both sides by $n^{-\delta}$, we obtain
\[
n^{-\delta}\log \hat{p}_S(\bm s^0|\bm \theta)=n^{-\delta}\log \hat{p}_G(\bm s^0|\bm \theta)+O\big\{ n^{-\delta}e^{-\hat{\gamma}_{\bm \theta}^{n}n^{\delta}\,(\bm s^0-\hat{\bm \mu}_{\bm \theta}^{n})^{T}\,\big(n^{\delta}\hat{\bm \Sigma}_{\bm \theta}^{n}\big)^{-1}\,(\bm s^0-\hat{\bm \mu}_{\bm \theta}^{n})}\big\},
\]
and, as $m$ and $n \rightarrow \infty$,  assumptions \ref{assConv}, \ref{assMaxEig} and \ref{assGamma} imply
\begin{equation} \label{eq:convNotNorm}
n^{-\delta}\log \hat{p}_S(\bm s^0|\bm \theta)\rightarrow n^{-\delta}\log \hat{p}_G(\bm s^0|\bm \theta),
\end{equation}
in probability, for any $\bm \theta$. Identifiability then follows from theorem \ref{theoGaus}. 


\end{appendices}

\pagebreak
\clearpage
\begin{center}
\textbf{\large Supplementary material to ``An Extended Empirical Saddlepoint Approximation for Intractable Likelihoods''\\}
\vspace{2mm}
Matteo Fasiolo, Simon N. Wood, Florian Hartig and Mark V. Bravington
\end{center}
\setcounter{section}{0}
\setcounter{equation}{0}
\setcounter{figure}{0}
\setcounter{table}{0}
\setcounter{page}{1}
\makeatletter
\renewcommand{\theequation}{S\arabic{equation}}
\renewcommand{\thefigure}{S\arabic{figure}}
\renewcommand{\bibnumfmt}[1]{[S#1]}
\renewcommand{\citenumfont}[1]{S#1}

\section{Asymptotics of the multivariate empirical saddlepoint approximation} \label{ref: asymEmpSad}

Here we follow \cite{feuerverger1989empirical} but develop the results in a multivariate setting, and with some changes in notation. For $ \bm \lambda \in I$, $\hat{M}_m( \bm \lambda)$ converges to $M( \bm \lambda)$ almost surely. This convergence is uniform and extends to $\hat{K}_m( \bm \lambda)$:
\begin{equation} \label{eq:Mconv}
\underset{ \bm \lambda \in I}{\text{sup}} \; | \hat{M}_m( \bm \lambda) - M( \bm \lambda) | \to 0,
\end{equation}
\begin{equation} \label{eq:Kconv}
\underset{ \bm \lambda \in I}{\text{sup}} \; | \hat{K}_m( \bm \lambda) - K( \bm \lambda) | \to 0.
\end{equation}

\underline{Proof}: Due to the Strong Law of Large Numbers $\hat{M}_m( \bm \lambda)$ converges to $M( \bm \lambda)$ almost surely, for all $ \bm \lambda$ in any countable collection $\{ \bm \lambda_i\}$. In addition $\hat{M}_m( \bm \lambda)$ and $M( \bm \lambda)$ are both convex functions and, for such functions, convergence on dense subsets implies uniform convergence on compact subsets \citep{robertsconvex}. This proves (\ref{eq:Mconv}), while (\ref{eq:Kconv}) follows by continuity of the logarithm. 

For $ \bm \lambda$ in the interior of $I$, these results extend to derivatives of both $\hat{M}_m( \bm \lambda)$ and $\hat{K}_m( \bm \lambda)$:
\begin{equation} \label{eq:M_deriv_conv}
\underset{ \bm \lambda \, \in \, \text{int}(I)}{\text{sup}} \; | D^{ i} \hat{M}_m( \bm \lambda) - D^{ i} M( \bm \lambda) | \to 0,
\end{equation}
\begin{equation} \label{eq:K_deriv_conv}
\underset{ \bm \lambda \, \in \, \text{int}(I)}{\text{sup}} \; | D^{ i} \hat{K}_m( \bm \lambda) - D^{ i} K( \bm \lambda) | \to 0,
\end{equation}
where $ i = \big \{i_1, \dots, i_d \big \}$ and:
$$
D^{ i} M( \bm \lambda) = \frac{\partial^k M( \bm \lambda)}
{\partial  \lambda_1^{i_1} \cdots \partial \lambda_d^{i_d}}, \;\;\;\;\; \text{with} \;\;\;\;
\sum_{z = 1}^K i_z = k \in {N}.
$$
\underline{Proof}: $D^{ i} M( \bm \lambda)$ is f{\text{}}inite only for $ \bm \lambda \in \text{int}(I)$. If all the elements of $ i$ are even, then $D^{ i} \hat{M}_m( \bm \lambda)$ and $D^{ i} M( \bm \lambda)$ are convex and (\ref{eq:M_deriv_conv}) follows as before. 
Otherwise, indicate with $ \bm \lambda^o$ the elements of $ \bm \lambda$ for which the corresponding element of $ i$ is odd. If there is an even number of components of $ \bm \lambda^o$ which are negative, $D^{ i} M( \bm \lambda)$ is still convex, otherwise $-D^{ i} M( \bm \lambda)$ is. Applying the uniform convergence argument for convex functions to the two sub-cases proves (\ref{eq:M_deriv_conv}). In addition, $D^{ i} K( \bm \lambda)$ has the form $P( \bm \lambda) /  M( \bm \lambda)  ^{2^k}$ 
with $P( \bm \lambda)$ being a polynomial function of $D^{ l} K( \bm \lambda)$, where $ l$ belongs to the set of all $d$-dimensional vector such that:
$$
l_j \in {N},  \;\;\;\; \sum_{j=1}^d{l_j} \leq k \;\;\;\; \text{for} \;\;\; j = 1, \dots, d.
$$
Given that an analogous argument holds for $D^{ i} \hat{K}_m( \bm \lambda)$, (\ref{eq:K_deriv_conv}) is proved by continuity.

After noticing that $\hat{M}_m( \bm \lambda)$ and its derivatives are unbiased estimators of $M( \bm \lambda)$ and its corresponding derivatives, it is straightforward to show that:
$$
m \,\text{Cov}\, \big \{ D^{ i} \hat{M}_m( \bm \lambda_1), D^{ j} \hat{M}_m( \bm \lambda_2)\big \} = 
D^{ i +  j} M( \bm \lambda_1+  \bm \lambda_2) - D^{ i} M( \bm \lambda_1) D^{ j} M( \bm \lambda_2),
$$
for $ \bm \lambda_1$, $ \bm \lambda_2$ such that $ \bm \lambda_1 +  \bm \lambda_2 \in I$.
This entails that, if we def{\text{}}ine $I/2$ to be the subset of $I$ such that $ \bm \lambda \in I/2$ if $2  \bm \lambda \in I$, than $\hat{M}_m( \bm \lambda)$ is a $\sqrt{m}$-consistent estimator of $M( \bm \lambda)$, for $ \bm \lambda \in I/2$. 
An analogous, but asymptotic, result for $\hat{K}_m( \bm \lambda)$ is the following:
$$
m \, \text{Cov}\, \big \{ D^{ i} \hat{K}_m( \bm \lambda_1), D^{ j} \hat{K}_m( \bm \lambda_2)\big \} \to D^{ i +  j} \bigg \{ \frac{ M( \bm \lambda_1 +  \bm \lambda_2)}{M( \bm \lambda_1)M( \bm \lambda_2)} - 1 \bigg \},
$$
where $ \bm \lambda_1$ and $ \bm \lambda_2$ are further restricted to the interior of $I/2$ if any of the elements of $ i$ or $ j$ is greater than zero. Finally, after noticing that on $I/2$:
$$
\hat{ \bm \lambda} = \hat{K'}^{-1}( \bm x) =  \bm \lambda + O(m^{-\frac{1}{2}}),
$$
we have that:
\begin{equation*}
\begin{array} {lcl} 
\frac{\hat{p}_m( \bm s)}{\hat{p}( \bm s)} & = & \frac{\text{det}\{K''( \bm \lambda)\}}{ \text{det}\{\hat{K}_m''(\hat{ \bm \lambda})\}} 
\text{exp}\Bigg [ \big \{ \hat{K}_m(\hat{ \bm \lambda}) - \hat{ \bm \lambda}^T\hat{K}'_m(\hat{ \bm \lambda}) \big \} - \big \{ K( \bm \lambda) -  \bm \lambda^T K'( \bm \lambda) \big \}     \Bigg ] \\ & = &  \frac{\text{det}\{K''( \bm \lambda)\}}{ \text{det}\{K''( \bm \lambda)\} + O(m^{-\frac{1}{2}}) }  \text{exp} \big \{ O(m^{-1/2}) \big \} \\ &=& 1 + O(m^{-\frac{1}{2}}),
\end{array}
\end{equation*}
by Taylor expansions, which are justif{\text{}}ied by the differentiability of all the functions involved. See \cite{feuerverger1989empirical} for more details.

\section{Optimality of the cross-validated Extended Empirical Saddlepoint} \label{app: KLoptim}

Let $p(\bm s|\bm \theta)$ be the true density of the statistics and
$\hat{p}_{S}(\bm s|\bm \theta, \gamma)$ be the EES density. Assume that we have a training set
of size $m$, a test set of size $n_{T}$ and that we have used $l$ simulations to normalize the density estimator. In this section we prove that, as $m$, $n_T$ and $l \rightarrow \infty$, Algorithm \ref{alCross} consistently selects the value of $\gamma$ which minimizes the Kullback-Leibler divergence between $\hat{p}_{S}(\bm s|\bm \theta, \gamma)$ and $p(\bm s|\bm \theta)$.  When two folds are
used, cross-validation (Algorithm \ref{alCross}) selects $\gamma$ as follows
$$
\hat{\gamma}=\text{\ensuremath{\underset{\gamma}{\text{argmin}}}}\bigg \{ -\frac{1}{n_{T}}\sum_{i=1}^{n_{T}}\log\hat{p}_{S}(\bm s_{i}|\bm \theta, \gamma) \bigg \} \;\;\;\;\text{with}\;\;\;\; \bm s_{i}\sim p(\bm s|\bm \theta),
$$
but the Weak Law of Large Numbers implies that
$$
\begin{array} {lcl} 
\underset{m, l , n_{T}\rightarrow\infty}{\text{plim}} -\frac{1}{n_{T}}\sum_{i=1}^{n_{T}}\log\hat{p}_{S}(\bm s_{i}|\bm \theta, \gamma)  & = & -\int \log p_{S}(\bm s|\bm \theta, \gamma)p(\bm s|\bm \theta)ds \\
& \propto & \int\big\{\log p(\bm s|\bm \theta)p(\bm s|\bm \theta)-\log p_{S}(\bm s|\bm \theta, \gamma)p(\bm s|\bm \theta)\big\} ds \\
& = & \int\log\frac{p(\bm s|\bm \theta)}{p_{S}(\bm s|\bm \theta, \gamma)}p(\bm s|\bm \theta)ds \\
& = & \text{KL}\bigg\{ p_{S}(\bm s|\bm \theta, \gamma),p(\bm s|\bm \theta)\bigg\}.
\end{array}
$$
Hence $p_{S}(\bm s|\bm \theta, \hat{\gamma})$ is the member of the $p_{S}(\bm s|\bm \theta, \gamma)$
family with minimal Kullback-Leibler distance from $p(\bm s|\bm \theta)$.
This result can easily be extended to $k$-fold cross-validation ($k > 2$).

\section{Practical implementation} \label{app: practicalImpl}

\subsection{Saddlepoint version of Algorithm \ref{alPoint}} \label{app: saddleAlg1}

In this section we illustrate how a pointwise synthetic likelihood estimate can be obtained using the new density estimator, rather than a Gaussian density.
\begin{algorithm} \label{alSaddle}
\caption{Estimating $p_{SL}({ \bm s}^0|{ \bm \theta})$ using the Extended Empirical Saddlepoint approximation}
\begin{algorithmic}[1]
\STATE Simulate datasets ${ \bm Y}_i, \dots, {\bm Y}_m$ from the model $p({ \bm Y}| \bm \theta)$.
\STATE Transform each dataset ${ \bm Y}_i$ to a vector of summary statistics $ \bm S_i = S({ \bm Y}_i)$.
\STATE Calculate sample mean $\hat{ \bm \mu}_{ \bm \theta}$ and covariance $\hat{ \bm \Sigma}_{ \bm \theta}$ of the simulated statistics. 
\STATE Estimate the synthetic likelihood
\begin{equation*} 
\hat{p}_{SL}({ \bm s}^0|{ \bm \theta}) = \hat{p}_m( \bm s^0, \gamma) = \frac{1}{  (2 \pi)^{\frac{d}{2}} \, \text{det}\{\hat{K}_m''( \hat{\bm \lambda}_m, \gamma, \bm s^0) \}^{\frac{1}{2}} } 
e^{\hat{K}_m(\hat{ \bm \lambda}_m, \gamma, \bm s^0) - \hat{ \bm \lambda}_m^T  \bm s^0},
\end{equation*} 
where $\hat{ \bm \lambda}_m$ is the solution of the empirical saddlepoint equation
\begin{equation*}
\hat{K}'_m(\hat{ \bm \lambda}_m, \gamma, \bm s^0) = \bm s^0,
\end{equation*}
while $\hat{K}_m(\bm \lambda, \gamma, \bm s)$ is given by equation (\ref{eq:modCGF}) in the main text.

\STATE  Optionally, normalize $\hat{p}_{SL}({ \bm s}^0|{ \bm \theta})$ by importance sampling
$$
\bar{p}_{SL}({ \bm s}^0|{ \bm \theta}) = 
\frac{\hat{p}_m( \bm s^0, \gamma)}{\hat{z}_m(\gamma)},
$$
where
$$
\hat{z}_m(\gamma) = \frac{1}{l} \sum_{i = 1}^l \frac{\hat{p}_m( \bm S_i, \gamma)}{q( \bm S_i)}, 
\;\;\;\;\;
\bm S_i \sim q( \bm s), \;\;\; \text{for} \;\; i = 1,\dots,l. 
$$
A reasonably eff{\text{}}icient importance density $q( \bm s)$ is a Gaussian density with mean vector $\hat{ \bm \mu}_{ \bm \theta}$ and covariance $\hat{ \bm \Sigma}_{ \bm \theta}$.
\end{algorithmic}
\end{algorithm}

\subsection{Maximizing the synthetic likelihood} \label{app: optimDetails}

To maximize the synthetic likelihood we have used a special case of the Iterated Filtering procedure, f{\text{}}irstly proposed by \cite{ionides2006inference}. Very brief{\text{}}ly, suppose that $\hat{ \bm \theta}_{k}$ is the estimate of the unknown parameters at the $k$-th step of the optimization routine. This estimate is updated as follows:
\begin{enumerate}
\item Simulate $N$ parameter vectors $ \bm \theta_1, \dots, \bm \theta_N$ from a user-def{\text{}}ined density $p( \bm \theta_{k+1}|\hat{ \bm \theta}_k)$ such that
\begin{equation} \label{kerCond}
\mathbb{E}( \bm \theta_{k+1}|\hat{ \bm \theta}_k) = \hat{ \bm \theta}_k, \;\;\; 
\text{var}( \bm \theta_{k+1}|\hat{ \bm \theta}_k) = \sigma^2_k  \bm \Sigma
\;\;\; \text{and} \;\;\; \mathbb{E}(|| \bm \theta_{k+1} - \hat{ \bm \theta}_k||^{3/2}) = o(\sigma_k^2),
\end{equation}
where $\sigma^2_k$ is a cooling schedule and $ \bm \Sigma$ is a covariance matrix.
\item For each $ \bm \theta_i$, obtain an estimate $\hat{p}_{SL}( \bm s^0| \bm \theta_i)$ of the synthetic likelihood, using either the multivariate normal density or the normalized EES.
\item Update the estimate
$$
\hat{ \bm \theta}_{k+1} = \frac{\sum_{i = 1}^N  \bm \theta_i \hat{p}_{SL}( \bm s^0| \bm \theta_i)}{
                               \sum_{i = 1}^N \hat{p}_{SL}( \bm s^0| \bm \theta_i)}.
$$
\end{enumerate}  
The convergence properties of this procedure have been studied, in the context of Hidden Markov Models, f{\text{}}irstly by \cite{ionides2006inference} and more in details by \cite{ionides2011iterated}. \cite{doucet2013derivative} explicitly pointed out that it can be used as a general likelihood optimizer. While those papers considered situations where the likelihood ($p_{SL}( \bm s^0| \bm \theta)$ in our context) can be evaluated exactly, we have verif{\text{}}ied empirically that the algorithm works well also when the likelihood is estimated with Monte Carlo error.
For all the examples we used the following cooling schedule
$$
\sigma_k^2 = \sigma_0^{2k}, \;\;\;\; \sigma_0^2 = 0.95. 
$$
In the shifted exponential example we performed 4 separate runs of the optimizer, using either the normal or the EES approximation, in both the 10 and 20-dimensional setting.

\subsection{Shifted exponential details} \label{sec: expDetails}

In one dimension, the ABC likelihood is
$$
p_{\epsilon}(s^0|\theta)=p(|s-s^0|<\epsilon|\theta)=\int_{0}^{\infty}I(|s-s^0|<\epsilon)p(s|\theta)ds,
$$
Without loss of generality, choose $s^0=0$. If $-\epsilon\leq\theta\leq\epsilon$
we have
$$
p_{\epsilon}(s^0|\theta)=\int_{\theta}^{\epsilon}p(s|\theta)ds=\int_{\theta}^{\epsilon}\beta e^{-\beta(s-\theta)}ds=\int_{0}^{\epsilon-\theta}\beta e^{-\beta x}dx=1-e^{-\beta(\epsilon-\theta)}=F(\epsilon-\theta),
$$
where we used the change of variable $x = s-\theta$.  Similarly, if $\theta<-\epsilon$, the likelihood is
$$
p_{\epsilon}(s^0|\theta)=\int_{-\epsilon}^{\epsilon}p(s|\theta)ds=F(\epsilon-\theta)-F(-\epsilon-\theta)=e^{-\beta(-\epsilon-\theta)}-e^{-\beta(\epsilon-\theta)}.
$$
Finally, $p_{\epsilon}(s^0|\theta)=0$ for $\theta>\epsilon$. Under a uniform prior on $[\psi, 0]$, where $\psi < -\epsilon$, the MAP is 
$$
\hat{\theta}=\underset{\theta}{\text{argmax}}\;p_{\epsilon}(s^0|\theta)=-\epsilon,
$$
with $p_{\epsilon}(s^0|\hat{\theta})=F(2\epsilon)=1-e^{-2\beta\epsilon}$.
In $d$ dimensions, the likelihood is
$$
p_{\epsilon}({\bm s}^0|\bm \theta) = \prod_{k=1}^d p_{\epsilon}(s^0_k|\theta_k),
$$
due to the independence between the summary statistics. Hence, the likelihood at the MLE is $F(2\epsilon)^{d}$, which is also the maximal probability that a simulated statistics vector gets accepted. 

In ABC the tolerance is often chosen so that a fraction $\alpha \in (0, 1)$ of the
statistics simulated from the prior falls within the tolerance. In one dimension and for fixed $\psi$, the overall probability of acceptance during this process is
$$
\begin{array}{lcl}
p(|S|<\epsilon|\psi) & = &  \int_{\psi}^{0}p(-\epsilon<S<\epsilon|\theta)\frac{1}{-\psi}d\theta \\
& = & -\frac{1}{\psi}\bigg\{\int_{\psi}^{-\epsilon}\big[e^{-\beta(-\epsilon-\theta)}-e^{-\beta(\epsilon-\theta)}\big]d\theta+\int_{-\epsilon}^{0}(1-e^{-\beta(\epsilon-\theta)})d\theta\bigg\} \\
& = & -\frac{1}{\psi}\bigg\{\frac{1}{\beta}\bigg[1+e^{\beta\psi}\bigg(e^{-\beta\epsilon}-e^{\beta\epsilon}\bigg)-e^{-\beta\epsilon}\bigg]+\epsilon\bigg\}.
\end{array}
$$
Now, to select $\psi$ so that we obtain an acceptance probability
equal to $\phi$, we need to solve 
$$
p(|S|<\epsilon|\psi) = \phi,
$$
wrt $\psi$, numerically (e.g. using bisection). Due to the independence between the priors and between the statistics, in $d$ dimensions the above probability becomes $p(|S|<\epsilon|\psi)^d$.

\subsection{Unstable population model details} \label{app: crashDetails}

Under the Gaussian and EES version of SL, we maximize the synthetic likelihood using 100 iterations of Iterated Filtering, with $N = 24$ synthetic likelihood evaluations at each step. The optimizer is initialized at $r=0.3$, $\kappa=30$, $\alpha=0.15$ and $\beta=0.03$. Given that $m = 5 \times 10^3$, estimating the model parameters costs $12 \times 10^6$ simulations from the model. In ABC we use $10^6$ simulation to calibrate the tolerance $\epsilon$, followed by $12 \times 10^6$ MCMC samples. Of these, we store only a thinned sub-sample of $22 \times 10^3$ parameter vector. Before obtaining MAP estimates, we discard the first $4 \times 10^3$ of these as burn-in. Then we maximize the posterior using the mean shift algorithm, where the approximate posterior is estimated using a Gaussian kernel density  estimator. Following the rule of thumb of \cite{silverman1986density} the covariance matrix, $\bf H$, of the kernels is determined as follows
$$
{\bf H}^{\frac{1}{2}} = \Big ( \frac{4}{d+2} \Big )^\frac{1}{d+4} n^{-\frac{1}{d+4}} \hat{\bm\Sigma}^{1/2},
$$
where $n=18 \times 10^3$, while ${\bf H}^{\frac{1}{2}}$ and ${\bm\Sigma}^{1/2}$ are matrix square roots of ${\bf H}$ and of $\hat{\bm\Sigma}$, the estimated covariance of the posterior samples. Given that the kernel density estimate of the posterior might have multiple local modes, we obtain 500 different MAP estimates by initializing the mean shift algorithm at a random posterior sample. We used the estimate corresponding to the highest (estimate) posterior density as our final MAP estimate.

\subsection{Formind settings} \label{app: formDetails}
The summary statistic were obtained using the following constants
$$
\alpha_{1, 1} = \alpha_{1, 3} = \alpha_{2, 1} = \alpha_{2, 3} = 1.5, \; \alpha_{1, 2} = 2, \alpha_{2, 2} = 2 \\
$$
while $\psi_{jk}$ and $\sigma_{jk}$ were estimates of mean and standard deviations of $C_{jk}$, obtained by simulating tree counts at the true parameters. The 24 datasets were simulated from Formind using the same parameter values as in Table 1 in the supplementary material of  \cite{hartig2014technical}. The chosen tree classes correspond to diameters at breast height $d < 0.2m$, $0.2m \leq d < 0.6m$, $d \geq 0.6m$ for pioneer and $d < 0.5m$, $0.5m \leq d < 1.4m$, $d \geq 1.4m$ for late successional trees. To generate the datasets the model was run for $10^5$ years, and the f{\text{}}inal statistics vector was selected. The $m = 10^4$ summary statistics simulated to estimate $p_{SL}( \bm s^0| \bm \theta)$ have been generated by simulating the model for $5.1 \times 10^4$ years, where the f{\text{}}irst $10^3$ years of simulation were discarded to avoid the transient, and by storing a vector of statistics every 5 years. 


Starting from initial values  $\mu_{pio} = 0.03$, $\mu_{suc} = 0.003$, $s_{pio} = 120$ and $s_{suc} = 40$, we ran the optimizations using $N=24$ and 100 iterations. The estimates reported in Table 1 in the main text were obtained by using the averages of the last 10 iterations of each optimization run as point estimates. The whole experiment took around 10 days on a quad-core Intel i7 3.6 GHz processor.

\section{Example: correlated multivariate shifted exponential distribution}

In the shifted exponential example included in the main text, the elements of the random vector $\bm S$ are independent. To show that EES can cope with correlated random variables we have introduced correlations, without altering the marginal densities, by using a copula model. In particular, we used a Gaussian $d$-dimensional copula, which has density
$$
c(u_1, \dots, u_d|\bm R) = \text{det}(\bm \Sigma)^{-\frac{1}{2}} \exp \bigg\{ \frac{1}{2}\bm q^T (\bm I_d - \bm R^{-1}) \bm q \bigg \},
$$
where $\bm R$ is a $d \times d$ correlation matrix, $\bm I_d$ is the identity matrix, $\bm q$ is a $d$-dimensional vector with $q_i=\Phi^{-1}(u^i)$, where $\Phi$ is the Cumulative Distribution Function of a standard normal. The random vector $\{ u_1, \dots, u_d \}$ has marginals that are uniformly distributed on $[0, 1]$. For an introduction to copulas see \cite{cherubini2004copula}.

To simulate unstructured, dense correlation matrices $\bm R$, we have used the method proposed by \cite{joe2006generating}. To set up the copula model and to simulate random variable, we have used the tools described by \cite{yan2007enjoy}.

We compare EES with a Gaussian and a kernel density estimator. In particular, we used $m = 10^3$ training samples and $5 \times 10^3$ test samples. The normalizing constant of the saddlepoint was estimated using $l = 10^3$ simulations and $\gamma$ was estimated by cross-validation. For the kernel estimator we used a multivariate Gaussian kernel with covariance $\alpha \hat{\bm \Sigma}$, where $\hat{\bm \Sigma}$ is the empirical covariance matrix of the random vectors in training set and $\alpha$ is a scaling parameter, whose value was selected by cross-validation. Figure \ref{fig:kern_dens} shows how the estimated Kullback-Leibler divergence, between the true density and each density estimate, varies with the number of dimensions. The true density is very skewed in each dimension, hence the Gaussian estimator is highly biased. The kernel estimator does better than the Gaussian, even as the dimensionality increases. This is attributable to the fact that having a single bandwidth $\alpha$ is very helpful in this example, because all the marginal densities are identical. The new density estimator performs uniformly better than the alternatives. 

\begin{figure} 
\centering
\includegraphics[scale = 0.51]{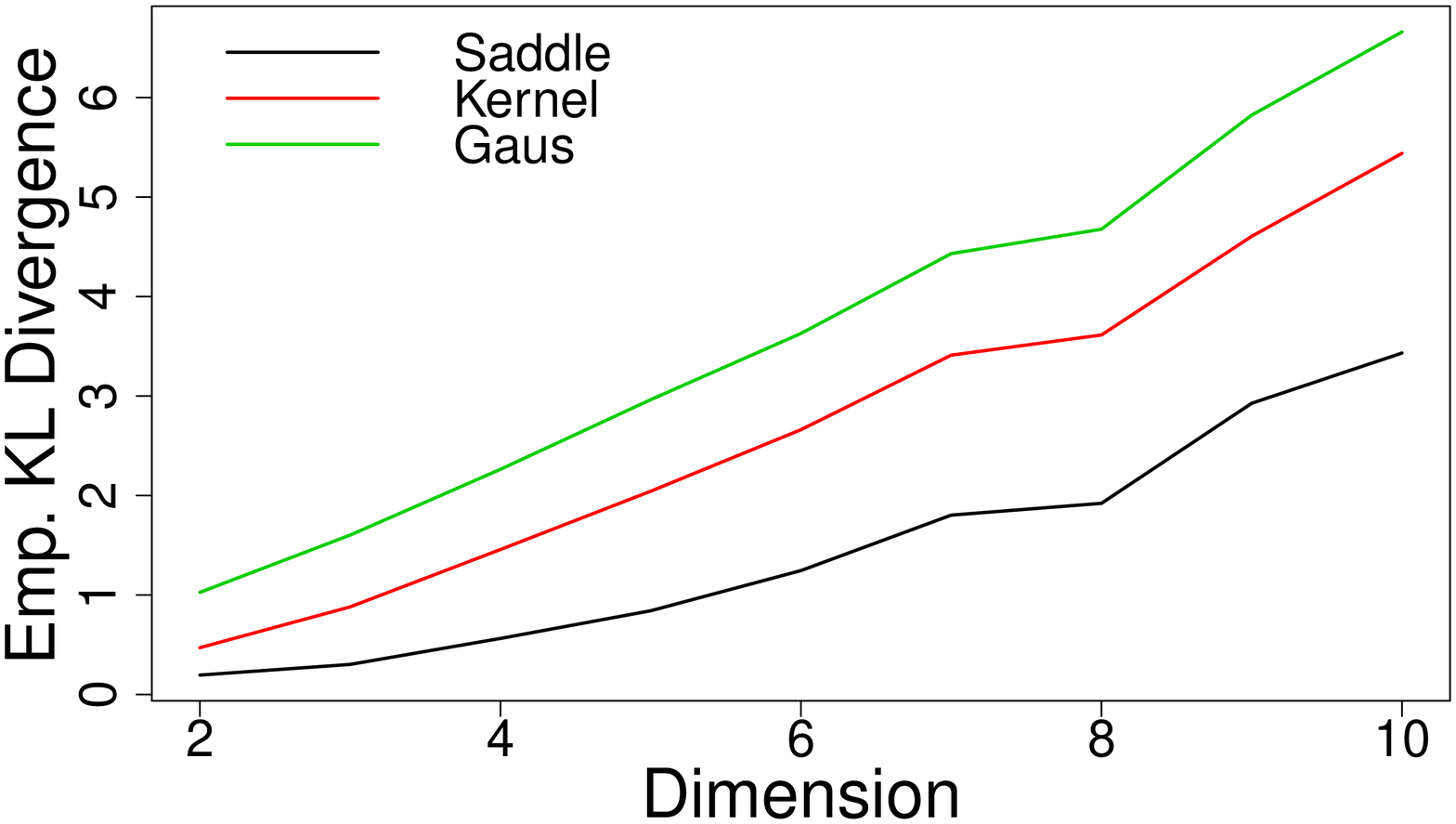}
\caption{Empirical Kullback-Leibler divergence between the three density estimators and the true density, as the number of dimensions increases.}
\label{fig:kern_dens}
\end{figure} 
 
As in the uncorrelated scenario (see the main text) we now estimate the shifts $\theta_1, \dots, \theta_d$, using the Gaussian and the new density estimator. We have considered a 10 and a 20-dimensional scenario. In both cases $\gamma$ has been selected by cross-validation. We have used $\beta = 0.2$ and $\theta_1 = \cdots = \theta_d = 0$. Given that the shape of the densities does not change with any of the $\theta$s we set $l=0$, and we have not computed the normalizing constant. We have used $m = 10^4$ and $m = 5 \times 10^4$ simulated vectors, respectively. By using EES, the Mean Squared Error was reduced from 21.9 to 4.8 in the 10-dimensional setting, and from 22.7 to 3.2 in the $20$-dimensional setting. P-values from t-test for differences in log-absolute errors were lower than $10^{-9}$ in both runs.

\end{document}